\documentclass[11pt,a4paper]{article}

\usepackage[a4paper,text={150mm,240mm},centering,headsep=10mm,footskip=15mm]{geometry}
\usepackage[T1]{fontenc}
\usepackage{epsfig, float,color}
\usepackage{amsmath,amssymb, amscd}
\usepackage{amsfonts,amsbsy}
\usepackage{hyperref}
\usepackage{url}
\usepackage{bbm}
\usepackage{mathrsfs} 
\usepackage{cite}

\newcommand{\R}{\mathbb{R}}
\newcommand{\C}{\mathbb{C}}
\newcommand{\N}{\mathbb{N}}

\newcommand{\vx}{{\mathbf{x}}}
\newcommand{\vy}{{\mathbf{y}}}
\newcommand{\vt}{{\mathbf{t}}}

\newcommand{\Banach}{\mathscr{B}}

\newcommand{\ret}{{\rm ret}}

\newcommand{\sym}{{\rm sym}}
\newcommand{\free}{{\rm free}}
\newcommand{\be}{\begin{equation}}
\newcommand{\ee}{\end{equation}}

\DeclareMathOperator*{\esssup}{ess \, sup}

\newtheorem{theorem}{Theorem}[section]

\newenvironment{proof}[1][Proof:]{\begin{trivlist}
\item[\hskip \labelsep {\bfseries #1}]}{\end{trivlist}}

\newenvironment{remark}[1][Remark:]{\begin{trivlist}
\item[\hskip \labelsep {\bfseries #1}]}{\end{trivlist}}

\newcommand{\qed}{\hfill\ensuremath{\square}}

\title{A new class of Volterra-type integral equations\\ from relativistic quantum physics}

\author{
Matthias Lienert\footnote{Department of Mathematics,
     Rutgers University,
     110 Frelinghuysen Road, Piscataway, NJ 08854-8019, USA.
     E-mail: lienertmat@gmail.com}\ \ and
Roderich Tumulka\footnote{Fachbereich Mathematik, Eberhard-Karls-Universit\"at, Auf der Morgenstelle 10, 72076 T\"ubingen, Germany. E-mail: roderich.tumulka@uni-tuebingen.de}
}

\date{July 24, 2018}

\begin{document}

\maketitle

  \begin{abstract}
  	\noindent Here we study a new kind of linear integral equations for a relativistic quantum-mechanical two-particle wave function $\psi(x_1,x_2)$, where $x_1,x_2$ are spacetime points. In the case of retarded interaction, these integral equations are of Volterra-type in the in the time variables, i.e., they involve a time integration from 0 to $t_i = x_i^0,~i=1,2$. They are interesting not only in view of their applications in physics, but also because of the following mathematical features: (a) time and space variables are more interrelated than in normal time-dependent problems, (b) the integral kernels are singular, and the structure of these singularities is non-trivial, (c) they feature time delay. We formulate a number of examples of such equations for scalar wave functions and prove existence and uniqueness of solutions for them. We also point out open mathematical problems.
	
  	\vspace{0.3cm}
  	
  \end{abstract}

\section{Introduction} \label{sec:intro}

The basic motivation of this paper is to present a new class of integral equations arising from relativistic quantum physics. The existence of solutions for these equations is not obvious, and we describe an approach that allows for proofs of existence and uniqueness results for simple examples from that class.

The straightforward extension of the concept of a quantum mechanical wave function to the relativistic case, due to Dirac \cite{dirac_32} (and in different form to Tomonaga \cite{tomonaga} and Schwinger \cite{schwinger}), uses wave functions that, for $N$ particles, are functions of $N$ space-time points, i.e.,
\be\label{psidef}
	\psi : \big(\R^{1,d} \big)^N \rightarrow \C^k,~~~(x_1,...,x_N) \mapsto \psi(x_1,...,x_N).
\ee
Here, $\R^{1,d}$ stands for (1+$d$)-dimensional Minkowski spacetime, $k \in \N$ depends on the type of particles described by $\psi$ and $x_i = (t_i,\vx_i) \in \R^{1,d}$ with $\vx_i \in \R^d$. Because of the occurrence of $N$ time coordinates $t_i$, $\psi$ has been termed a \textit{multi-time wave function} \cite{dice_paper}. The usual single-time wave function $\varphi(t,\vx_1,...,\vx_N)$ is straightforwardly contained in $\psi$ as the special case of equal times, i.e., $\varphi(t,\vx_1,...,\vx_N) = \psi(t,\vx_1,...,t,\vx_N)$, while $\psi$ is a manifestly Lorentz-covariant object.

Apart from being needed for Lorentz invariance, the $N$ time coordinates also make a new type of evolution equation possible that includes direct interactions between the particles. Since interaction cannot occur faster than light in a relativistic setting, these interactions need to have a retarded effect, with a time delay proportional to the distance; that is, the interaction should take place along light cones, as in the Wheeler--Feynman formulation of classical electromagnetism \cite{wf2}. As detailed in \cite{direct_interaction_quantum}, starting from the reformulation of the usual Schr\"odinger equation for $N=2$ as an integral equation, a natural generalization of the equation to the relativistic case leads to the following class of \textit{multi-time integral equations:}
\be
	\psi(x_1,x_2) = \psi^\free(x_1,x_2) + \lambda \int dV(x_1') \int dV(x_2') \, G_1(x_1-x_1') G_2(x_2-x_2') K(x_1',x_2') \psi(x_1',x_2').
	\label{eq:inteq}
\ee
Here, $\psi^\free(x_1,x_2)$ is a given solution of free (i.e., non-interacting) relativistic wave equations such as the Klein-Gordon (KG) equation or the Dirac equation in each spacetime variable $x_1, x_2$. We shall focus on the KG equation for which we have
\be
	(\square_i + m_i^2)\psi^\free(x_1,x_2) = 0,~~~i=1,2.
	\label{eq:freekgmultitime}
\ee
Here, $\square_i = \partial_{t_i}^2 - \Delta_i$ denotes the wave operator acting on $x_i$. Furthermore, $G_1, G_2$ are Green's functions of these equations, $\lambda \in \R$ is a coupling constant, $dV(x_i),~i=1,2$ are the (1+$d$)-dimensional spacetime volume elements, the integrals run over $\R^{1,d}$ and $K(x_1,x_2)$ is the so-called \textit{interaction kernel.} 

Eq.~\eqref{eq:inteq} defines the class of integral equations that this paper is about. We shall give more details in Sec.~\ref{sec:inteq} and refer to \cite{direct_interaction_quantum} for details about physical background and motivation.

Another source of motivation for studying \eqref{eq:inteq} is that similar equations can heuristically be derived from quantum field theory (QFT) for the description of bound states of two particles. In fact, the well-known Bethe-Salpeter (BS) equation \cite{bs_equation,greiner_qed} is of the form \eqref{eq:inteq} with a distribution-valued $K$.

So far, relativistic two-particle wave functions have almost exclusively been considered (a) in the non-interacting case when the task reduces to solving well-known free equations such as \eqref{eq:freekgmultitime}, and (b) in the interacting case with recourse to QFT (see \cite{schweber}). Even if one is only interested in the two-particle wave function $\psi(x_1,x_2)$, the QFT dynamics nevertheless involves a Fock space function, i.e., a collection of $n$-particle wave functions $\psi^{(n)}(x_1,...,x_n)$ for every $n \in \N$. Eq. \eqref{eq:inteq}, by contrast, only involves one function of eight variables instead of a set of infinitely many functions, each of $4n$ variables. Even more importantly, QFT typically faces the problem of ultra-violet divergences. That means, some of the expressions involved in QFT do not make sense when taken literally but are divergent \cite{folland}. As we shall demonstrate here, Eq. \eqref{eq:inteq}, by contrast, makes sense as it stands and leads to a well-posed initial value problem. Besides, our proof of the existence and uniqueness of solutions is based on an iteration scheme that might serve as the basis of numerical algorithms. Moreover, also the feature that \eqref{eq:inteq} allows to express \textit{direct interactions with time delay} is new compared to the existing approaches (apart from the BS equation, see \cite{direct_interaction_quantum} for a discussion).

In the mathematical literature, integral equations of the form \eqref{eq:inteq} have, to the best of our knowledge, not been systematically analyzed before.\footnote{Note that several works in the physics literature on the BS equation study (special) solutions of that equation (see e.g. \cite{wick_54,cutkosky54,green57,tiktopoulos65,consenza65,obrien75} and references therein). However, these works seem to be of limited significance for the mathematical theory of \eqref{eq:inteq}, for the following reasons. Several of the works use a Wick rotation \cite{wick_54}, i.e., they replace Minkowski spacetime with (1+$d$)-dimensional Euclidean space. This greatly simplifies the equation, as Green's functions of the equation $(-\Delta_{1+d} + m^2)\varphi = 0$, where $\Delta_d$ denotes the Laplacian in $d$ dimensions, are much simpler than Green's functions of the Klein-Gordon equation $(\partial_t^2 - \Delta_d + m^2)\varphi = 0$. However, a transformation back to the original equation is not attempted (and may well not always be possible). In addition, some of these works set $\psi^\free = 0$. The same step here would lead to only the trivial solution $\psi=0$. Moreover, these works study an eigenvalue problem of the form $\psi = \lambda \widehat{L} \psi$, where $\widehat{L}$ is an integral operator, in the coupling constant $\lambda$. As the physical value of the coupling constant is fixed, the results are only indirectly relevant for the actual problem. Lastly, the question of suitable initial data (or a different classification of solutions) is left untouched. As a consequence, these works have nothing to say about how to understand the BS equation as a law for defining the time evolution of $\psi$.} Several points make the task challenging:
\begin{enumerate}
	\item \textit{Non-trivial time dependence.} Because of the structure $G_1(x_1-x_1') G_2(x_2-x_2') K(x_1',x_2')$, integral transformations in the two time coordinates do not render the problem simple. In particular, the problem cannot easily be reduced to a time-independent one.
	\item \textit{Infinite domains.} The domain of integration occurring in \eqref{eq:inteq} is $\R^{1,d}\times \R^{1,d}$. That means, in order for the integral to exist, the integrand needs to have a certain drop-off behavior at infinity, e.g., has to be in $L^1$. However, as the Green's functions of the relevant wave equations do not decay particularly fast, $\psi$ needs to provide this drop-off behavior. This is problematic, as the integral operator then maps out of $L^1$, and it becomes hard to even set up a suitable mathematical framework. (We shall illustrate this problem in detail in Sec. \ref{sec:problem}.)
	\item \textit{Combined singularities of the Green's functions.} Green's functions of relativistic wave equations are typically singular; for example, they often contain Dirac $\delta$-functions on the light cone. If in addition $K$ is singular as well, which is the case for the physically natural choice in 1+3 dimensions, $K(x_1,x_2)=\delta((t_1-t_2)^2-|\vx_1-\vx_2|^2)$ \cite{direct_interaction_quantum}, then the structure of the combined singularities of $G_1, G_2$ and $K$ becomes particularly difficult to treat.
\end{enumerate}

We shall address these problems as follows. We focus on the case that $G_1,G_2$ are Green's functions of the Klein-Gordon equation. In order to formulate a tractable class of models, we set aside item 3 (returning to it later) and assume that $K$ is bounded. At least in $d=1$, this assumption also turns out to be physically realistic. With regard to item 2, we note that the infinite domain $\R^{1,d}\times \R^{1,d}$ is not the only physically reasonable possibility. Cosmologists take seriously the possibility that our universe had a Big Bang and thus is only semi-infinite in time. To keep the discussion simple, we implement this beginning in time in a rather crude way, cutting off Minkowski spacetime before $t=0$. (The case of curved cosmological spacetimes which actually feature a Big Bang singularity is studied in a separate paper \cite{int_eq_curved}.)
Focusing on the case of \emph{retarded} Green's functions (i.e., $G(x_1-x_1')$ that are nonzero only for $x_1'$ on or within the past light cone of $x_1$; a common choice in physics with reference to causality) then renders the time integrals finite, and leads to a Volterra-type structure of the equations. In fact, the whole domain of integration in \eqref{eq:inteq} becomes finite.
These simplifications then permit us to deal with item 1, the non-trivial time dependence.\\

The paper is structured as follows.
First, we formulate the integral equation in full detail on Minkowski spacetime for the relevant space dimensions $d=1,2,3$ (Sec. \ref{sec:inteq}). At the example of $d=1$, we illustrate the above-mentioned problem of infinite domains (Sec. \ref{sec:problem}). This motivates us to formulate the integral equation on semi-infinite (1+$d$)-dimensional Minkowski spacetime (Sec. \ref{sec:simplifiedmodel}).

Sec. \ref{sec:results} is dedicated to the question of the existence and uniqueness of solutions of the identified models. In Sec. \ref{sec:banachspaces} we point out which Banach spaces seem appropriate for the physical problem. Then, in Sec. \ref{sec:l2kernels}, we connect to standard results about multi-time Volterra integral equations by proving an existence and uniqueness theorem (Thm. \ref{thm:l2kernels}) for arbitrary $L^\infty L^2$-kernels ($L^\infty$ in the times and $L^2$ in the space variables). The proof is rather standard but serves to recall classical arguments and to prepare the method of the later proofs. Noting that the kernels in our models are not $L^\infty L^2$-kernels, we turn to the case of realistic Green's functions $G$ and bounded interaction kernels $K(x_1,x_2)$. (The total integral kernel is still singular.)  Section \ref{sec:boundedkernels} contains our main results: existence and uniqueness theorems for the simplified models of Sec. \ref{sec:simplifiedmodel} (Thms. \ref{thm:1dboundedkernel}-\ref{thm:3dboundedkernel}). The proofs make crucial use of the fact that the retarded Green's functions of relativistic wave equations are only supported on (and possibly within) past light cones. In Sec. \ref{sec:singularkernels}, we finally extend our method to certain (special) singular interaction kernels which are simplified compared to the physically natural cases (Thm.s \ref{thm:singularkernel3d}, \ref{thm:singularkernel2d}).

In Sec. \ref{sec:conclusion}, we conclude and put our results in perspective. Moreover, we point out open problems that may be of interest to researchers specializing in integral equations.

\section{The integral equation} \label{sec:inteq}

We now make explicit the physically relevant form of our integral equation~\eqref{eq:inteq} in $d=1,2,3$ space dimensions.

\subsection{Explicit form of the Green's functions and the integral equation} \label{sec:explicitform}

The integral equation \eqref{eq:inteq} becomes fully specified by the choices of $G_1, G_2$ and $K$. Here we focus on the case that $\psi$ is complex-valued [i.e., $k=1$ in \eqref{psidef}] and $G_1, G_2$ are retarded Green's functions of the Klein-Gordon (KG) equation, i.e., for $j=1,2$,
\be
	(\square_j + m_j^2)G_j^\ret(x_j) = \delta^{(1+d)}(x_j),
\ee
and $G_j^\ret(t_j,\vx_j) = 0$ for $t_j <0$. Here, $m_j \geq 0$ is the $j$-th particle's mass and $\delta^{(1+d)}$ denotes the (1+$d$)-dimensional delta function.

In dimensions $d=1,2,3$, the retarded Green's functions with mass $m=m_1$ or $m=m_2$ are given as follows (see \cite[chap. 7.4]{zauderer}, \cite[appendix E]{birula_qed}). We use the abbreviation $x^2 = (x^0)^2-|\vx|^2$ for $x = (x^0,\vx) \in \R^{1,d}$ (Minkowski square) and set the physical constants $c$ and $\hbar$ to unity.
\begin{itemize}
	\item[] $d=1$: $G^\ret(x) = \frac{1}{2} H(x^0) H(x^2) J_0(m \sqrt{x^2})$,
	\item[] $d=2$: $G^\ret(x) = \frac{1}{2\pi} H(x^0) H(x^2) \frac{\cos(m\sqrt{x^2})}{\sqrt{x^2}}$,
	\item[] $d=3$: $G^\ret(x) = \frac{1}{2\pi} H(x^0) \delta(x^2) - \frac{m}{4\pi \sqrt{x^2}} H(x^0) H(x^2) J_1(m\sqrt{x^2})$.
\end{itemize}
Here, $H(s)=1_{s>0}$ denotes the Heaviside function and $J_0, J_1$ are Bessel functions of the first kind of order 0 and 1, respectively.

As detailed in \cite{direct_interaction_quantum}, the physically natural choice of the interaction kernel is  $K(x_1,x_2) = G^\sym(x_1-x_2)$, the time-symmetric Green's function of the wave equation (i.e., the massless KG equation). We have:
\begin{itemize}
	\item[] $d=1$: $G^\sym(x) = \frac{1}{2} H(x^2)$,
	\item[] $d=2$: $G^\sym(x) = \frac{1}{2\pi} \frac{H(x^2)}{\sqrt{x^2}}$,
	\item[] $d=3$: $G^\sym(x) = \frac{1}{2\pi} \delta(x^2)$.
\end{itemize}
With these choices, the integral equation \eqref{eq:inteq} in the various space dimensions becomes:

\paragraph{d=1:}
\begin{align}
	&\psi(t_1,z_1,t_2,z_2) = \psi^\free(t_1,z_1,t_2,z_2) + \frac{\lambda}{8} \int_{-\infty}^{t_1} dt_1' \int dz_1' \int_{-\infty}^{t_2} dt_2' \int dz_2'~ H(t_1-t_1'-|z_1-z_1'|)\nonumber\\
&~\times~ J_0\Big(m_1\sqrt{(t_1-t_1')^2-|z_1-z_1'|^2}\Big) \, H(t_2-t_2'-|z_2-z_2'|) \, J_0\Big(m_2\sqrt{(t_2-t_2')^2-|z_2-z_2'|^2}\Big)\nonumber\\
&~\times~ H((t_1'-t_2')^2-|z_1'-z_2'|^2) \, \psi(t_1',z_1',t_2',z_2').
\label{eq:inteq1d}
\end{align}

\paragraph{d=2:}
\begin{align}
	&\psi(t_1,\vx_1,t_2,\vx_2) = \psi^\free(t_1,\vx_1,t_2,\vx_2) + \frac{\lambda}{(2\pi)^3} \int_{-\infty}^{t_1} dt_1' \int d^2 \vx_1' \int_{-\infty}^{t_2} dt_2' \int d^2 \vx_2'~\nonumber\\
& ~\times~  H(t_1-t_1'-|\vx_1-\vx_1'|)\, \frac{\cos\big(m_1\sqrt{(t_1-t_1')^2-|\vx_1-\vx_1'|^2}\big)}{\sqrt{(t_1-t_1')^2-|\vx_1-\vx_1'|^2}} \,
H(t_2-t_2'-|\vx_2-\vx_2'|) \, \nonumber\\
&~\times~  \frac{\cos\big(m_2\sqrt{(t_2-t_2')^2-|\vx_2-\vx_2'|^2}\big)}{\sqrt{(t_2-t_2')^2-|\vx_2-\vx_2'|^2}}\, \frac{H((t_1'-t_2')^2-|\vx_1'-\vx_2'|^2) }{\sqrt{(t_1'-t_2')^2-|\vx_1'-\vx_2'|^2}}\, \psi(t_1',\vx_1',t_2',\vx_2').
\label{eq:inteq2d}
\end{align}

\paragraph{d=3:} For simplicity, we consider only the massless case ($m_1=m_2=0$). Then the most singular terms of the Green's functions are still included, and the equation takes the form
\begin{align}
	&\psi(t_1,\vx_1,t_2,\vx_2) = \psi^\free(t_1,\vx_1,t_2,\vx_2) + \frac{2\lambda}{(4\pi)^3} \int_{-\infty}^{t_1} dt_1' \int d^3 \vx_1' \int_{-\infty}^{t_2} dt_2' \int d^3 \vx_2'~  \nonumber\\
& ~\times~\frac{\delta(t_1-t_1'-|\vx_1-\vx_1'|)}{|\vx_1-\vx_1'|}\, \frac{\delta(t_2-t_2'-|\vx_2-\vx_2'|)}{|\vx_2-\vx_2'|} \, \delta((t_1'-t_2')^2-|\vx_1'-\vx_2'|^2) \, \psi(t_1',\vx_1',t_2',\vx_2').
\label{eq:inteq3d}
\end{align}

The form of the equations for the different dimensions is quite different, both with respect to the domain of integration and with respect to the singularities that occur. In $d=1$ and $d=2$, the domain of integration are the time-like configurations, i.e., the set
\be
	\mathcal{T} = \{ (t_1,\vx_1,t_2,\vx_2) \in \R^{1,d} \times \R^{1,d} : |t_1-t_2| > |\vx_1-\vx_2|\}.
\ee
In $d=3$, however, because of the delta function $\delta((t_1'-t_2')^2-|\vx_1'-\vx_2'|^2)$ in the interaction kernel, the integral in \eqref{eq:inteq3d} runs only along the light-like configurations,
\be
	\mathscr{L} =  \{ (t_1,\vx_1,t_2,\vx_2) \in \R^{1,d} \times \R^{1,d} : |t_1-t_2| = |\vx_1-\vx_2|\}.
\ee
As noted in \cite{direct_interaction_quantum}, \eqref{eq:inteq3d} can be solved on $\mathscr{L}$ autonomously. For a configuration outside of $\mathscr{L}$, it can then be used as a formula to calculate the solution.

Concerning the singularities, there are only jump singularities in $d=1$ and the whole integral kernel remains bounded. In $d=2$, there are three connected singularities of the form $1/\sqrt{t^2-\vx^2}$. Finally, in $d=3$, there are singularities of the form $1/|\vx|$ and also $\delta$-functions which require some care to be defined rigorously, and which may lead to further singularities because of the weight factor associated with the roots of their arguments.

These connected singularities may be quite hard to treat in $d=2, 3$. However, because in all cases the domains extend infinitely in the time direction, there is a more basic problem we have to deal with first. We shall illustrate this in the case $d=1$ where the singularities are unproblematic.

\subsection{Difficulties with infinite time integrations} \label{sec:problem}

Consider the integral equation \eqref{eq:inteq1d} in $d=1$. The well-posedness of the problem at the very least requires the integral to exist. As $|J_0| \leq 1$ and $J_0(0) = 1$, the existence of the integral in the massless case implies the existence for every $m_1,m_2>0$. So we focus on $m_1=m_2=0$. Then we obtain the following condition on $\psi$:
\begin{multline}
	(\widehat L \psi)(t_1,z_1,t_1',z_1'):= \int_{-\infty}^{t_1} dt_1' \int dz_1' \int_{-\infty}^{t_2} dt_2' \int dz_2'~ H(t_1-t_1'-|z_1-z_1'|) \\[2mm]
\times~ H(t_2-t_2'-|z_2-z_2'|) \, H((t_1'-t_2')^2-|z_1'-z_2'|^2)\, |\psi|(t_1',z_1',t_2',z_2') < \infty
	\label{eq:integral1d}
\end{multline}
for all $t_1,t_2, z_1,z_2$. This means, $\psi$ needs to have certain integrability properties which are related to its behavior for $t_1,t_2 \rightarrow \pm \infty$. As only configurations $(t_1',z_1',t_2',z_2') \in \mathcal{T}$ contribute to the integral, a natural possibility is to demand $\psi \in L^1(\mathcal{T})$. Then the integral \eqref{eq:integral1d} is finite.

However, in order to formulate the equation \eqref{eq:inteq1d} on $\psi \in L^1(\mathcal{T})$, we also need that the integral operator $\widehat{L}$ occurring in the equation maps from $L^1(\mathcal{T})$ to $L^1(\mathcal{T})$. This yields the further condition:
\begin{align}
	&\hspace{-5mm}\int dt_1 \int_{-\infty}^{t_1} dt_1' \int dt_2 \int_{-\infty}^{t_2} dt_2'  \int dz_1\, dz_1' \, dz_2 \, dz_2' ~ H(t_1-t_1'-|z_1-z_1'|) \, H(t_2-t_2'-|z_2-z_2'|)\nonumber\\
\times~ &H((t_1-t_2)^2-|z_1-z_2|^2)\, H((t_1'-t_2')^2-|z_1'-z_2'|^2)\, |\psi|(t_1',z_1',t_2',z_2') < \infty
	\label{eq:doubleintegral1d}
\end{align}
for all $\psi \in L^1(\mathcal{T})$.
The point now is that the integral \eqref{eq:doubleintegral1d} simply diverges. This can, for example, be seen from the fact that arbitrarily large $t_1,t_2$ contribute to the integral.
Thus, it seems difficult to even start the mathematical analysis of \eqref{eq:inteq1d}. A similar problem also occurs in the massive case and in dimensions $d=2,3$ as well.

While we do not claim that it is generally impossible to analyze the integral equations on domains which are infinite in the times, we have not found any other way to deal with the problem besides modifying the equation. The root of the problem obviously lies in the fact that the domain of integration is infinite in time. In the retarded case, the domain extends to $-\infty$ instead of stopping at some finite value. The easiest remedy is to assume that spacetime does not extend back to $t\to -\infty$ but had an initial time which thus becomes a lower bound of the time integrations. This renders the integral \eqref{eq:integral1d} finite without demanding some kind of drop-off behavior of $\psi$ in time (e.g., for $\psi \in L^\infty(\R^4)$). Of course, the cutoff requires a justification from physics, as it breaks important symmetries such as time translation invariance (and also Lorentz invariance).

Fortunately, there is such a physical justification. Cosmology has come to the conclusion that it is not unlikely that our universe has a Big Bang singularity, i.e., a beginning in time. To implement the Big Bang properly requires to formulate the integral equation \eqref{eq:inteq} on curved spacetimes. This is a non-trivial task by itself and is the topic of a separate paper \cite{int_eq_curved}. Among other things, one needs to explicitly determine the Green's functions of the appropriate quantum mechanical wave equations on the respective curved spacetimes. Here we set aside these complications, content ourselves that there is a physical reason for a beginning in time, and simply cut off the time integrals in \eqref{eq:inteq} at $t_1=t_2=0$.

\subsection{Simplified models} \label{sec:simplifiedmodel}

The cutoff at $t_1=t_2=0$ gets rid of the problem of infinite time domains. However, there is another problem remaining (for $d=2,3$): the connected singularities of $G_1,G_2$ and $K$. We shall deal with this problem as follows. The Green's functions $G_1,G_2$ cannot be modified as they are determined by the type of quantum mechanical particle under consideration. The interaction kernel $K$ is more arbitrary. There is a most natural choice for physics, $K(x_1,x_2) = G^\sym(x_1-x_2)$, but other choices just lead to a different kind of interaction. In particular, we can approximate $G^\sym(x_1-x_2)$ by a regular function while respecting the physical symmetries. For example, for $d=3$, $G^\sym(x_1-x_2) = \frac{1}{2\pi} \delta((x_1-x_2)^2)$ can be approximated by
\be
	K(x_1,x_2) = \frac{1}{ (2\pi)^{3/2} \sigma} \exp\left(- \frac{|(x_1-x_2)^2|^2}{2 \sigma^2} \right)
\ee
for some small constant $\sigma >0$.

In the following, we shall therefore study models with arbitrary but regular (e.g., bounded or bounded and smooth) interaction kernels $K$. We return to the case of singular interaction kernels in Sec. \ref{sec:singularkernels}.

The simplified models we shall study are given as follows.
\paragraph{d=1:} Here, the natural interaction kernel is bounded. We nevertheless allow for arbitrary bounded $K(x_1,x_2):$
\begin{align}
	&\psi(t_1,z_1,t_2,z_2) = \psi^\free(t_1,z_1,t_2,z_2) + \frac{\lambda}{4} \int_0^{t_1} dt_1' \int_0^{t_2} dt_2' \int dz_1' \, dz_2'~ H(t_1-t_1'-|z_1-z_1'|)\nonumber\\
&~\times~ J_0\Big(m_1\sqrt{(t_1-t_1')^2-|z_1-z_1'|^2}\Big) \, H(t_2-t_2'-|z_2-z_2'|) \, J_0\Big(m_2\sqrt{(t_2-t_2')^2-|z_2-z_2'|^2}\Big)\nonumber\\[1mm]
&~\times~ K(t_1',z_1',t_2',z_2') \, \psi(t_1',z_1',t_2',z_2').
\label{eq:inteq1dsimplified}
\end{align}

\paragraph{d=2:}
\begin{align}
	&\psi(t_1,\vx_1,t_2,\vx_2) = \psi^\free(t_1,\vx_1,t_2,\vx_2) + \frac{\lambda}{(2\pi)^2} \int_0^{t_1} dt_1' \int_0^{t_2} dt_2' \int d^2 \vx_1'\, d^2 \vx_2'~\nonumber\\
& ~\times~ H(t_1-t_1'-|\vx_1-\vx_1'|)\, \frac{\cos\big(m_1\sqrt{(t_1-t_1')^2-|\vx_1-\vx_1'|^2}\big)}{\sqrt{(t_1-t_1')^2-|\vx_1-\vx_1'|^2}} \,
H(t_2-t_2'-|\vx_2-\vx_2'|) \, \nonumber\\
&~\times~ \frac{\cos\big(m_2\sqrt{(t_2-t_2')^2-|\vx_2-\vx_2'|^2}\big)}{\sqrt{(t_2-t_2')^2-|\vx_2-\vx_2'|^2}}\, K(t_1',\vx_1',t_2',\vx_2') \,  \psi(t_1',\vx_1',t_2',\vx_2').
\label{eq:inteq2dsimplified}
\end{align}

\paragraph{d=3:} We again consider the case $m_1=m_2=0$ here and let $K(x_1,x_2)$ be smooth and bounded. If $\psi$ is a test function, this permits us to perform the time integrals by using the $\delta$-functions. We are left with:
\begin{multline}
	\hspace{-3mm}\psi(t_1,\vx_1,t_2,\vx_2) = \psi^\free(t_1,\vx_1,t_2,\vx_2) + \frac{\lambda}{(4\pi)^2}  \int d^3 \vx_1' \, d^3 \vx_2'~  \frac{H(t_1-|\vx_1-\vx_1'|)}{|\vx_1-\vx_1'|}\, \frac{H(t_2-|\vx_2-\vx_2'|)}{|\vx_2-\vx_2'|} \\[2mm]
 \times~ K(t_1-|\vx_1-\vx_1'|,\vx_1',t_2-|\vx_2-\vx_2'|,\vx_2') \,  \psi(t_1-|\vx_1-\vx_1'|,\vx_1',t_2-|\vx_2-\vx_2'|,\vx_2').
\label{eq:inteq3dsimplified}
\end{multline}
The Heaviside functions result from the lower bounds $t_1',t_2' \geq 0$ in the modified integral equation \eqref{eq:inteq3d}. In order to avoid complications with distribution-valued kernels, we directly study \eqref{eq:inteq3dsimplified}.

\paragraph{Remarks.}
\begin{enumerate}
	\item For all three equations \eqref{eq:inteq1dsimplified}-\eqref{eq:inteq3dsimplified}, the domain of integration is now effectively finite, as it is finite in the time directions and as the Green's functions $G_1^\ret$, $G_2^\ret$ are only supported along (and possibly inside) the backward light cones ${\rm PLC}(x_i) = \{ y \in \R^{1,d} : |x_i^0-y^0| > |\vx_i-\vy| ~{\rm and}~ y^0 < x^0_i\}$.
	\item For $d=1$ and $d=2$, the time integrals run from $0$ to $t_i$, $i=1,2$. That means, \eqref{eq:inteq1dsimplified} and \eqref{eq:inteq2dsimplified} have a multi-dimensional Volterra-type structure. We shall therefore call these equations \textit{multi-time Volterra integral equations} (MTVE). For $d=3$, the structure is somewhat different. However, we can also see that the radial coordinates $|\vx_i|$ can only take values between $0$ and $t_i$. This will also allow us to employ methods for Volterra integral equations for $d=3$.
	\item From the integral equations \eqref{eq:inteq1dsimplified}-\eqref{eq:inteq3dsimplified} we can read off that $\psi$ satisfies the initial value problem $\psi(0,\vx_1,0,\vx_2) = \psi^\free(0,\vx_1,0,\vx_2)$. In other words, $\psi$ is subject to a \textit{Cauchy problem ``at the Big Bang.''} As $\psi^\free$ is a solution of the free multi-time equations, here $(\square_k + m^2_k) \psi = 0,~k=1,2$, it is itself determined uniquely by Cauchy data. Thus, if we can prove the existence and uniqueness of solutions for arbitrary $\psi^\free$, we also obtain a classification of the solutions by Cauchy data at $t_1 = 0 = t_2$. For a multi-time integral equation \eqref{eq:inteq} on a domain which has no lower bounds in the times, the relation between $\psi^\free$ and initial values for $\psi$ is not as clear (see the discussion in \cite{direct_interaction_quantum}).
\end{enumerate}

\section{Results} \label{sec:results}

In the following, we prove a number of existence and uniqueness theorems for MTVEs. In Sec. \ref{sec:banachspaces} we discuss which Banach spaces seem appropriate for physics. In Sec. \ref{sec:l2kernels}, we pick up work in the literature on multi-dimensional Volterra integral equations and prove an existence and uniqueness result for general $L^\infty_{t_1,t_2} L^2_{\vx_1,\vx_2}$ kernels. We claim no originality for this result; however, the proof is useful to connect with classical results and to introduce the strategy of the following proofs. In Sec. \ref{sec:boundedkernels}, we turn to the existence and uniqueness proofs for Eqs. \eqref{eq:inteq1dsimplified}-\eqref{eq:inteq3dsimplified} with bounded interaction kernels. These theorems constitute our main results; they are not special cases of the general theorem in Sec. \ref{sec:l2kernels}. In Sec. \ref{sec:singularkernels}, we finally return to the case of singular interaction kernels, and show that the method developed in Sec. \ref{sec:boundedkernels} is sufficient to at least treat certain singular interaction kernels.

\subsection{Banach space} \label{sec:banachspaces}

We shall consider the integral equations \eqref{eq:inteq1dsimplified}-\eqref{eq:inteq3dsimplified} as linear operator equations on a suitable Banach space $\Banach$:
\be
    \psi = \psi^\free + \widehat{L}\psi,
\ee
where $\widehat{L}$ is the integral operator occurring in the respective equation.

Which space should $\Banach$ be? As $\psi$ is a quantum-mechanical wave function, there are some expectations about $\Banach$. In non-relativistic quantum mechanics, the single-time wave function $\varphi(t,\vx_1,\vx_2)$ represents a probability amplitude; hence, it has to be square integrable in the space variables $\vx_1,\vx_2$ for every fixed time $t$, and the $L^2$-norm is constant in time. This suggests choosing $\Banach$ to be the following Bochner space:
\be
    \Banach_d := L^\infty \big([0,T]^2_{(t_1,t_2)}, L^2(\R^{2d}_{(\vx_1,\vx_2)}) \big),
    \label{eq:banach}
\ee
where $T>0$ is an arbitrary constant, with the norm of $\psi \in \Banach_d$ given by
\be
    \| \psi\|_{\Banach_d} = \esssup_{t_1,t_2 \in [0,T]} \| \psi(t_1,\cdot,t_2,\cdot) \|_{L^2} \:.
    \label{eq:norm}
\ee
In fact, for solutions $\psi$ of \eqref{eq:inteq} or of the free Klein-Gordon equation, $|\psi|^2$ cannot be expected to represent a probability density, nor $\| \psi(t_1,\cdot,t_2,\cdot) \|_{L^2}$ to be independent of $t_1$ and $t_2$. Still, the choice \eqref{eq:banach} will turn out useful for our proofs.

\subsection{General $L^\infty_{t_1,t_2} L^2_{\vx_1,\vx_2}$-kernels} \label{sec:l2kernels}

Multi-dimensional Volterra integral equations have been treated in the literature before, see e.g. \cite{beesack_1984,beesack_1985}. These references cover equations of the form
\be
	f(\vt) = f_0(\vt) + \int_0^{\vt} d \vt'~L(\vt,\vt')f(\vt'),
	\label{eq:beesack}
\ee
	where $\vt = (t_1,...,t_N)$,
\be
	\int_0^{\vt} d\vt' = \int_0^{t_1}dt_1' \cdots \int_0^{t_N} dt_N'\,,
\ee
and the kernel $L$ is assumed to be either bounded or square integrable.
Our MTVEs \eqref{eq:inteq1dsimplified}-\eqref{eq:inteq3dsimplified} differ from \eqref{eq:beesack} in the following aspects.
\begin{enumerate}
	\item In addition to the time integrals, they also include space integrals. Space and time directions are distinguished by the form of the equations. Most importantly, the integral from 0 to $t$, which characterizes Volterra integral equations, only appears in the time directions.
	\item The kernels in Eqs.~\eqref{eq:inteq1dsimplified}-\eqref{eq:inteq3dsimplified} are in general not square integrable (see the Remark at the end of this section for details), the ones of Eqs.~\eqref{eq:inteq2dsimplified} and \eqref{eq:inteq3dsimplified} are in general not bounded either. Likewise, the specific kernels of Eqs.~\eqref{eq:inteq1d}-\eqref{eq:inteq3d} are not square integrable, those of \eqref{eq:inteq2d} and \eqref{eq:inteq3d} are not bounded.
\end{enumerate}
The first point can easily be approached using classical methods; we shall prove a corresponding theorem below (Thm. \ref{thm:l2kernels}). However, the second item shows that this is not enough to cover even the simplified physically relevant cases. It turns out that we need to utilize the more specific structure of the kernels of Eqs.~\eqref{eq:inteq1dsimplified}-\eqref{eq:inteq3dsimplified}. This will be done in Sec.~\ref{sec:boundedkernels}.\\

So let us describe how square-integrable kernels can be treated.
In the remainder of the section, we study the MTVE
\be
	f(\vt,\vx) = f_0(\vt,\vx) + \int_0^{\vt} d \vt' \int d\vx' ~L(\vt,\vx;\vt',\vx') f(\vt',\vx'),
	\label{eq:mtve}
\ee
where $\vt \in \R^N$, $\vx \in \R^M$. The integral equations \eqref{eq:inteq1dsimplified}, \eqref{eq:inteq2dsimplified} in $d=1$ and $d=2$ correspond to this structure for $N=2$ and $M=2$ or $M=4$ with special (but not square-integrable) integral kernels. The integral equation \eqref{eq:inteq3dsimplified} in $d=3$ is different because of the time shifts occurring in the integral.

\begin{theorem} \label{thm:l2kernels}
	Let $T>0$, consider the Banach space $\Banach = L^\infty \big([0,T]^N_{\vt}, L^2(\R^M_{\vx}) \big)$, and let
	\be
		\| L \|^2 = \esssup_{\vt, \vt' \in [0,T]^N} \int d \vx \, d\vx'~ |L|^2(\vt,\vx;\vt',\vx') < \infty.
	\label{eq:norml}
	\ee
	Then, for every $f_0 \in \Banach$, \eqref{eq:mtve} has a unique solution $f \in \Banach$.
\end{theorem} 

The proof serves as a good illustration of the basic technique that shall also be used in Sec. \ref{sec:boundedkernels}. It is based on classical methods for Volterra integral equations (see \cite[chap. 3.1]{linz} and \cite{beesack_1984,beesack_1985}).

\begin{proof}
	Let $f_0 \in \Banach$. The idea is to first show that the iterations
	\be
		f_n(\vt,\vx)  = f_0(\vt,\vx) +  \int_0^{\vt} d \vt' \int d\vx' ~L(\vt,\vx;\vt',\vx') f_{n-1}(\vt',\vx'),~~~n\in \N
	\ee
	converge.  In a second step, we demonstrate that the limiting function is indeed a solution of \eqref{eq:mtve}. Third, we show that the solution is unique.

	For convenience, we introduce
	\be
		\varphi_n = f_n - f_{n-1},~~~n\in \N
	\ee
	and $\varphi_0 = f_0$. We then have
	\be
		f_n = \sum_{i=0}^n \varphi_i
		\label{eq:varphiseries}
	\ee
	and the functions $\varphi_n$ satisfy the equation
	\be
		\varphi_n(\vt,\vx)  = \int_0^{\vt} d \vt' \int d\vx' ~L(\vt,\vx;\vt',\vx') \: \varphi_{n-1}(\vt',\vx'),~~~n\in \N.
		\label{eq:varphieq}
	\ee

Let $\widehat{L}$ denote the integral operator in \eqref{eq:mtve}. First of all, we show that $\widehat{L}$ is a bounded operator on $\Banach$. Then it follows in particular that $\varphi_n \in \Banach \, \forall n \in \N_0$.
So let $f \in \Banach$. Then $f$ is an equivalence class of functions. We choose an arbitrary representative of this class, a function on $[0,T]^N \times \R^M$ that is square-integrable in $\vx$ for almost every $\vt$, and call this function simply $f$ again. Using \eqref{eq:varphieq} and the Cauchy-Schwarz inequality, we find for every $n \in \N$ and $\vt\in [0,T]^N$:
\begin{align}
	\| (\widehat{L}f)(\vt,\cdot) \|^2_{L^2} &\leq \int d \vx \left[ \left( \int_0^\vt d\vt' \int d \vx' ~|L|^2(\vt,\vx;\vt',\vx')  \right) \left( \int_0^\vt d\vt' \int d \vx' ~|f|^2(\vt',\vx')\right) \right]\nonumber\\
	&= \left( \int_0^\vt d\vt' \int d \vx \, d\vx'~|L|^2(\vt,\vx;\vt',\vx') \right) \left( \int_0^\vt d\vt'~\| f(\vt',\cdot) \|^2_{L^2}\right)\nonumber\\
	&\leq (t_1 \cdots t_N) \: \|L\|^2 \:  \int_0^\vt d\vt'~\| f(\vt',\cdot) \|^2_{L^2},
	\label{eq:auxiliaryineq}
\end{align}
where it remains open at first whether the $L^2$ norms are finite or infinite. However, since $\| f(\vt',\cdot) \|_{L^2}\leq \|f\|_{\Banach}$ for almost every $\vt'$, we obtain that
\begin{align}
\| (\widehat{L}f)(\vt,\cdot) \|^2_{L^2} 
&\leq (t_1 \cdots t_N) \: \|L\|^2 \:  \int_0^\vt d\vt'~\|f\|^2_{\Banach} \nonumber\\
&\leq (t_1 \cdots t_N)^2 \: \|L\|^2 \: \|f\|^2_{\Banach} \label{eq:aux2ineq}
\end{align}
for every $\vt$, which is independent of the choice of representative of $f$. (In particular, the $L^2$ norm of $(\widehat{L}f)(\vt,\cdot)$ turns out finite for every $\vt$.)\\
Furthermore, the estimate \eqref{eq:aux2ineq} implies:
\be
	\| \widehat{L} f \|_{\Banach} = \esssup_{\vt \in [0,T]^N} \| \widehat{L} f (\vt,\cdot)\|_{L^2} \leq T^{N} \: \| L \| \: \| f \|_{\Banach}.
\ee
So $\widehat{L}$ is indeed a bounded operator on $\Banach$, and we have $\varphi_n \in \Banach \,\forall n \in \N_0$.

Next, we show that the sequence \eqref{eq:varphiseries} has a limit in $\Banach$. To this end, we now prove the following bound for the point-wise norms $\| \varphi_n(\vt,\cdot) \|_{L^2}$ by induction over $n \in \N_0$:
\be
	\| \varphi_n(\vt,\cdot) \|^2_{L^2} \leq \| f_0\|^2_{\Banach} \: \| L \|^{2n} \frac{(t_1 \cdots t_N)^{2n}}{(n!)^N}
	\label{eq:varphiind}
\ee
for every $\vt$. For $n=0$, the claim is obvious as $\varphi_0 = f_0$. So let \eqref{eq:varphiind} hold for some $n \in \N$. Recall $\varphi_{n+1} = \widehat{L} \varphi_n$. That means, we can use \eqref{eq:auxiliaryineq} to estimate the norm $\| \varphi_{n+1}(\vt,\cdot) \|^2_{L^2}$ in terms of $\| \varphi_n(\vt,\cdot) \|^2_{L^2}$. Plugging \eqref{eq:varphiind} for $n$ into \eqref{eq:auxiliaryineq} yields:
\begin{align}
	\| \varphi_{n+1}(\vt,\cdot) \|^2_{L^2} &\leq (t_1 \cdots t_N) \: \|L\|^2 \:  \int_0^\vt d\vt'~\| f_0\|^2_{\Banach} \: \| L \|^{2n} \frac{(t'_1 \cdots t'_N)^{2n}}{(n!)^N}\nonumber\\
	&= \| f_0\|^2_{\Banach} \: \| L \|^{2(n+1)} \frac{(t_1 \cdots t_N)^{2(n+1)}}{(n!)^N (2n+1)^N}\nonumber\\
	&\leq  \| f_0\|^2_{\Banach} \: \| L \|^{2(n+1)} \frac{(t_1 \cdots t_N)^{2(n+1)}}{[(n+1)!]^N}.
\end{align}
This proves \eqref{eq:varphiind}. In particular, \eqref{eq:varphiind} implies:
\be
	\| \varphi_n \|_{\Banach} \leq \| f_0\|_{\Banach} \: \| L \|^n \frac{T^{n N}}{(n!)^{N/2}} \: .
	\label{eq:normvarphin}
\ee
This bound in turn shows that the series $\sum_{i=0}^\infty \| \varphi_i \|_{\Banach}$ converges. Hence, the iterations converge, i.e., $f_n \rightarrow f \in \Banach$ for $n\rightarrow \infty$. This completes the first step of the proof.

Next, we show that the series $f= \sum_{i=0}^\infty \varphi_i$ indeed is a solution of \eqref{eq:mtve}. Since $\widehat L$ is bounded, we have that
\be
	\widehat{L} \sum_{i=0}^\infty \varphi_i = \sum_{i=0}^\infty \widehat{L}\varphi_i = \sum_{i=0}^\infty \varphi_{i+1} = \sum_{i=0}^\infty \varphi_i - \varphi_0,
	\label{eq:sumintexchange}
\ee
which is equivalent to $f = f_0 + \widehat{L} f$.

Finally, we turn to the uniqueness of the solution. To this end, let $\widetilde{f} \in \Banach$ be another solution of \eqref{eq:mtve}. Then the difference $g = f - \widetilde{f}$ satisfies the equation $g = \widehat{L} g$. This is similar to the equation $\varphi_n = \widehat{L} \varphi_{n-1}$. Thus, in the same way as we derived \eqref{eq:normvarphin}, we obtain the inequality
\be
	\| g\|_{\Banach} \leq \| g \|_{\Banach} \: \| L \|^n \frac{T^{n N}}{(n!)^{N/2}}~~\forall n \in \N.
\ee
Thus, for $n \rightarrow \infty$ we find $\|g \| = 0$, hence $f = \widetilde{f}$. This shows the uniqueness of the solution, completing the proof. \qed
\end{proof}

\paragraph{Remarks.}
\begin{enumerate}
\item	The kernels in the integral equations \eqref{eq:inteq1dsimplified}-\eqref{eq:inteq3dsimplified} may [and those of \eqref{eq:inteq1d}-\eqref{eq:inteq3d} do] violate the square-integrability condition \eqref{eq:norml} for the following reason: if $L$ is invariant under translations as in $(\vx,\vx') \mapsto (\vx+\mathbf{a},\vx'+\mathbf{a})$ for every $\mathbf{a}\in\R^d$, then $L(\vt,\cdot,\vt',\cdot)$, if nonzero, cannot be square-integrable in $\R^{2d}$ as a function of $\vx$ and $\vx'$ for any $\vt$ and $\vt'$.

\item Since we have obtained the existence of a unique solution up to time $T$ for arbitrary $T>0$, it follows that a unique solution exists for all times, i.e., on $[0,\infty)^N \times \R^M$. In fact, the estimate \eqref{eq:varphiind} shows that the solution can at most grow exponentially according to
\be
\|f(\vt,\cdot)\|_{L^2} \leq \|f_0\|_{\Banach} \: e^{\|L\| t_1 \cdots t_N} \,.
\ee
Thus, $e^{-\|L\| t_1\cdots t_N} f \in L^\infty([0,\infty)^N, L^2(\R^M))$ whenever $f_0$ lies in this space.
\end{enumerate}

\subsection{Bounded interaction kernels} \label{sec:boundedkernels}

In this section, we provide existence and uniqueness results for the simplified equations \eqref{eq:inteq1dsimplified}-\eqref{eq:inteq3dsimplified} and bounded interaction kernels $K$ (the overall kernel $L$ is still singular). The proofs use the same strategy as before, however, essential use is made of the fact that the equations only contain integrals along (and inside of) past light cones. This is the essential feature that allows us to deal also with kernels which do not satisfy \eqref{eq:norml}.

\begin{theorem}[$d=1$] \label{thm:1dboundedkernel}
	For every $m_1,m_2 \geq 0$, every
	\[
		\psi^\free \in \Banach_1 = L^\infty \big([0,T]^2_{(t_1,t_2)}, L^2(\R^{2}_{(z_1,z_2)}) \big), 
	\]
	and every essentially bounded $K : \R^4 \rightarrow \C$, the integral equation \eqref{eq:inteq1dsimplified} has a unique solution $\psi \in \Banach_1$.
\end{theorem}
The theorem evidently covers the physically natural interaction kernel in Eq. \eqref{eq:inteq1d}.

For the proof (and the following ones), we follow the strategy of the proof of Thm. \ref{thm:l2kernels}. We again define the functions $f_n$ as the $n$-th iteration (starting from $f_0 = \psi^\free$) and $\varphi_n$ as the difference $f_n-f_{n-1}$ with $\varphi_0 = f_0$. The integral operator in \eqref{eq:inteq1dsimplified} is denoted by $\widehat{L}$. We describe only the essential new steps in the proof (in particular how to obtain an appropriate estimate for $\| \varphi_n \|$); the rest then follows as before.

\begin{proof}
	We first show that $\widehat{L}$ maps $\Banach_1$ to $\Banach_1$. Using  \eqref{eq:inteq1dsimplified}, the Cauchy-Schwarz inequality, and that $|J_0| \leq 1$, we find:
\begin{align}
	&\| (\widehat{L}\psi)(t_1,\cdot,t_2,\cdot) \|^2_{L^2} \leq \frac{\lambda^2}{16} \int dz_1 \, dz_2 \biggl[ \biggl( \int_0^{t_1} dt_1' \int_0^{t_2} dt_2' \int d z_1' \, dz_2' ~H(t_1-t_1'-|z_1-z_1'|)  \nonumber\\
& \times H(t_2-t_2'-|z_2-z_2'|) \, |K|^2(t_1',z_1',t_2',z_2')  \biggr)\nonumber\\
& \times \left. \left( \int_0^{t_1} dt_1' \int_0^{t_2} dt_2' \int d z_1' \, dz_2' \, H(t_1-t_1'-|z_1-z_1'|) \, H(t_2-t_2'-|z_2-z_2'|) \, |\psi|^2(t_1',z_1',t_2',z_2') \right)\right]\nonumber\\
&\leq \frac{\lambda^2}{16} \int dz_1 \, dz_2 \left[ \left( \int_0^{t_1} dt_1' \int_0^{t_2} dt_2' ~ 2(t_1-t_1') \, 2(t_2-t_2') \, \| K\|_\infty^2 \right) \right.\nonumber\\
 & \times \left. \left( \int_0^{t_1} dt_1' \int_0^{t_2} dt_2' \int d z_1' \, dz_2' \, H(t_1-t_1'-|z_1-z_1'|) \, H(t_2-t_2'-|z_2-z_2'|) \, |\psi|^2(t_1',z_1',t_2',z_2') \right)\right]\nonumber\\
&= \frac{\lambda^2}{16} \, \| K\|_\infty^2 \, (t_1 t_2)^2 \int dz_1 \, dz_2 \int_0^{t_1} dt_1' \int_0^{t_2} dt_2' \int d z_1' \, dz_2'~ H(t_1-t_1'-|z_1-z_1'|) \nonumber\\
&~~~~~~~~~~~~~~~~~~~~\times H(t_2-t_2'-|z_2-z_2'|) |\psi|^2(t_1',z_1',t_2',z_2') \: .
\end{align}
Exchanging the order of integrations and performing the $z_1,z_2$-integrations first leads to:
\begin{align}
	&\| (\widehat{L} \psi)(t_1,\cdot,t_2,\cdot) \|^2_{L^2}   \nonumber\\
	&\leq \frac{\lambda^2}{16} \| K\|_\infty^2 \, (t_1 t_2)^2  \int_0^{t_1} dt_1' \int_0^{t_2} dt_2' \int d z_1' \, dz_2'~ 2(t_1-t_1') \, 2(t_2-t_2')
\, |\psi|^2(t_1',z_1',t_2',z_2')\nonumber\\
&= \frac{\lambda^2}{4}\, \| K\|_\infty^2 \, (t_1 t_2)^2  \int_0^{t_1} dt_1' \int_0^{t_2} dt_2' ~ (t_1-t_1') \, (t_2-t_2') \, \|\psi(t_1',\cdot,t_2'\cdot)\|^2_{L^2}
\label{eq:hilfsformelind1dl2}
\end{align}
Therefore, replacing $\|\psi(t_1',\cdot,t_2'\cdot)\|^2_{L^2}$ with $\|\psi \|_{\Banach_1}^2 = \esssup_{t'_1, t'_2 \in [0,T]} \|\psi(t_1',\cdot,t_2'\cdot)\|^2_{L^2}$, we find:
\be
	\| \widehat{L} \psi\|^2_{\Banach_1} \leq \frac{\lambda^2}{16}\esssup_{t_1, t_2 \in [0,T]} \| K\|_\infty^2 \, (t_1 t_2)^4 \, \|\psi\|_{\Banach_1}^2 = \frac{\lambda^2}{16} \, \| K\|_\infty^2 \, T^8 \, \|\psi\|_{\Banach_1}^2.
\ee
This shows that $\widehat{L} : \Banach_1 \to \Banach_1$ is bounded. In particular, $\varphi_n \in \Banach_1$ for all $n \in \N_0$. Next, we prove by induction that \footnote{The estimate is not optimal but sufficient.}
\be
		\| \varphi_n(t_1,\cdot,t_2,\cdot) \|^2_{L^2} \leq \left( \frac{\lambda^2}{4}\right)^{n} \: \| \psi^\free\|_{\Banach_1}^2 \: \| K \|_\infty^{2n} \: \frac{(t_1 t_2)^{4n}}{[(2n)!]^2} \,.
\label{eq:inf1dl2ind}
\ee
For $n=0$ this is obviously true. So let \eqref{eq:inf1dl2ind} hold for some $n \in \N_0$. Then, by plugging  \eqref{eq:inf1dl2ind} into \eqref{eq:hilfsformelind1dl2} we obtain that
\begin{align}
& \| \varphi_{n+1}(t_1,\cdot,t_2,\cdot) \|^2_{L^2}\nonumber\\
& \leq \left( \frac{\lambda^2}{4}\right)^{n+1} \| \psi^\free\|_{\Banach_1}^2  \: \| K\|_\infty^{2(n+1)} (t_1 t_2)^2  \int_0^{t_1} dt_1' \int_0^{t_2} dt_2' ~ (t_1-t_1')(t_2-t_2') \frac{(t_1' t_2')^{4n}}{[(2n)!]^2}\nonumber\\
&= \left( \frac{\lambda^2}{4}\right)^{n+1} \| \psi^\free\|_{\Banach_1}^2  \: \| K\|_\infty^{2(n+1)} \frac{(t_1 t_2)^{4(n+1)}}{[(2n)!]^2} \, \frac{1}{[(4n+1)(4n+2)]^2}\nonumber\\
& \leq \left( \frac{\lambda^2}{4}\right)^{n+1} \| \psi^\free\|_{\Banach_1}^2  \: \| K\|_\infty^{2(n+1)} \frac{(t_1 t_2)^{4(n+1)}}{[(2(n+1))!]^2}\:.
\end{align}
This proves \eqref{eq:inf1dl2ind}. In particular, \eqref{eq:inf1dl2ind} implies:
\be
	\| \varphi_n \|_{\Banach_1} \leq \left(\frac{|\lambda|}{2}\right)^n  \| \psi^\free\|_{\Banach_1} \: \| K \|_\infty^n \, \frac{T^{4n}}{(2n)!} \: .
\ee
Hence, $\sum_i \| \varphi_i \|_{\Banach_1} < \infty$. Going through the analogous steps as in the proof of Thm. \ref{thm:l2kernels}, we find that $\sum_i \varphi_i \in \Banach_1$ yields the unique solution of \eqref{eq:inteq1dsimplified}. \qed
\end{proof}

\begin{theorem}[$d=2$] \label{thm:2dboundedkernel}
	For every $m_1,m_2 \geq 0$, every essentially bounded $K : \R^6 \rightarrow \C$ and every $\psi^\free \in \Banach_2= L^\infty \big([0,T]^2_{(t_1,t_2)}, L^2(\R^{4}_{(\vx_1,\vx_2)}) \big)$, \eqref{eq:inteq2dsimplified} has a unique solution $\psi \in \Banach_2$.
\end{theorem}

The proof again uses the previous ideas and notation.

\begin{proof}
	We first show that the integral operator $\widehat{L}$ in \eqref{eq:inteq2dsimplified} is a bounded operator on $\Banach_2$. Using \eqref{eq:inteq2dsimplified} and the Cauchy-Schwarz inequality, we find for every $\psi \in \Banach_2$:
\begin{align}
	& \| (\widehat{L} \psi)(t_1,\cdot,t_2,\cdot) \|^2_{L^2}\nonumber\\
 & \leq \frac{\lambda^2}{(2\pi)^4} \int d^2 \vx_1 \, d^2 \vx_2 \left[ \left(  \int_0^{t_1} dt_1'  \int_0^{t_2} dt_2' \int d^2 \vx_1' \, d^2 \vx_2'~ \frac{H(t_1-t_1'-|\vx_1-\vx_1'|)}{\sqrt{(t_1-t_1')^2 - |\vx_1-\vx_1'|^2}} \right. \right.\nonumber\\
& \left. \frac{H(t_2-t_2'-|\vx_2-\vx_2'|)}{\sqrt{(t_2-t_2')^2 - |\vx_2-\vx_2'|^2}} |K|^2(t_1',\vx_1',t_2',\vx_2') \right) \left(  \int_0^{t_1} dt_1'  \int_0^{t_2} dt_2' \int d^2 \vx_1' \, d^2 \vx_2'~ \right.\nonumber\\
& \left. \left. \frac{H(t_1-t_1'-|\vx_1-\vx_1'|)}{\sqrt{(t_1-t_1')^2 - |\vx_1-\vx_1'|^2}}  \frac{H(t_2-t_2'-|\vx_2-\vx_2'|)}{\sqrt{(t_2-t_2')^2 - |\vx_2-\vx_2'|^2}} |\psi|^2(t_1',\vx_1',t_2',\vx_2') \right) \right].
\label{eq:2dinfcalc1}
\end{align}
The expression in the first round brackets is smaller than or equal to
\begin{align}
	&\| K \|_\infty^2 \int_0^{t_1} dt_1'  \int_0^{t_2} dt_2' \int d^2 \vx_1' \, d^2 \vx_2'~ \frac{H(t_1-t_1'-|\vx_1-\vx_1'|)}{\sqrt{(t_1-t_1')^2 - |\vx_1-\vx_1'|^2}} \frac{H(t_2-t_2'-|\vx_2-\vx_2'|)}{\sqrt{(t_2-t_2')^2 - |\vx_2-\vx_2'|^2}}\nonumber\\
&= \| K \|_\infty^2 \int_0^{t_1} dt_1'  \int_0^{t_2} dt_2' ~ (2\pi)^2 (t_1-t_1') (t_2-t_2')\nonumber\\
&= (2\pi)^2 \| K \|_\infty^2 \frac{(t_1t_2)^2}{4},
\end{align}
where in the second line we made use of the identity
\be
	\int_{|\vx| < \tau}d^2 \vx \, \frac{1}{\sqrt{\tau^2-|\vx|^2}} = 2 \pi \tau.
	\label{eq:identity2d}
\ee
Thus, we find with \eqref{eq:2dinfcalc1}:
\begin{align}
	&\| (\widehat{L} \psi)(t_1,\cdot,t_2,\cdot) \|^2_{L^2} \leq \frac{\lambda^2}{(2\pi)^2} \| K \|_\infty^2 \frac{(t_1t_2)^2}{4} \int d^2 \vx_1 \, d^2 \vx_2 \int_0^{t_1} dt_1'  \int_0^{t_2} dt_2' \int d^2 \vx_1' \, d^2 \vx_2' \nonumber\\
 & ~~~~\frac{H(t_1-t_1'-|\vx_1-\vx_1'|)}{\sqrt{(t_1-t_1')^2 - |\vx_1-\vx_1'|^2}}  \frac{H(t_2-t_2'-|\vx_2-\vx_2'|)}{\sqrt{(t_2-t_2')^2 - |\vx_2-\vx_2'|^2}} \, |\psi|^2(t_1',\vx_1',t_2',\vx_2')
\end{align}
Now we change the order of integration, introduce new integration variables $\vy_i = \vx_i-\vx_i'$ instead of $\vx_i$ and integrate over $\vy_i$ (using again \eqref{eq:identity2d}). This yields:
\begin{align}
	& \| (\widehat{L} \psi)(t_1,\cdot,t_2,\cdot) \|^2_{L^2} \nonumber\\
&\leq \lambda^2 \, \| K \|_\infty^2 \frac{(t_1t_2)^2}{4} \int_0^{t_1} dt_1'  \int_0^{t_2} dt_2' \int d^2 \vx_1' \, d^2 \vx_2' ~(t_1-t_1') (t_2-t_2') \,  |\psi|^2(t_1',\vx_1',t_2',\vx_2')\nonumber\\
&= \lambda^2 \, \| K \|_\infty^2 \frac{(t_1t_2)^2}{4} \int_0^{t_1} dt_1'  \int_0^{t_2} dt_2' ~(t_1-t_1') (t_2-t_2') \,  \| \psi(t_1',\cdot,t_2',\cdot) \|^2_{L^2}.
\label{eq:hilfsformelind2dl2}
\end{align}
Replacing $\| \psi(t_1',\cdot,t_2',\cdot) \|^2_{L^2}$ in \eqref{eq:hilfsformelind2dl2} with $\| \psi \|^2_{\Banach_2}$, we find:
\be
	\| \widehat{L} \psi \|^2_{\Banach_2} \leq \esssup_{t_1,t_2 \in [0,T]} \lambda^2 \, \| K \|_\infty^2 \frac{(t_1t_2)^2}{4} \, \| \psi \|^2_{\Banach_2} \, \frac{(t_1 t_2)^2}{4} = \lambda^2 \, \| \psi \|^2_{\Banach_2} \, \| K \|_\infty^2 \, \frac{T^8}{16}.
\ee
This shows that $\widehat{L}$ is a bounded operator on $\Banach_2$. Hence, $\varphi_n \in \Banach_2 \, \forall n \in \N_0$.

Next, we prove the following estimate for $\varphi_n,~n\in \N_0$:
\be
	\| \varphi_n(t_1,\cdot,t_2,\cdot) \|_{L^2}^2 \leq \lambda^{2n} \, \| \psi^\free \|^2_{\Banach_2} \, \frac{\| K \|_\infty^{2n}}{4^n} \, \frac{(t_1t_2)^{4n}}{[(2n)!]^2} \:.
	\label{eq:2dinfind}
\ee
For $n=0$ this obviously holds. So let \eqref{eq:2dinfind} be true for some $n \in \N_0$. Plugging \eqref{eq:2dinfind} into \eqref{eq:hilfsformelind2dl2} yields:
\begin{align}
	\| \varphi_{n+1}(t_1,\cdot,t_2,\cdot) \|^2_{L^2} &\leq \lambda^{2(n+1)} \, \| \psi^\free \|^2_{\Banach_2} \, \frac{\| K \|_\infty^{2(n+1)}}{4^{n+1}} \int_0^{t_1} dt_1'  \int_0^{t_2} dt_2' ~(t_1-t_1') (t_2-t_2') \,  \frac{(t_1't_2')^{4n}}{[(2n)!]^2}\nonumber\\
 & =  \lambda^{2(n+1)} \, \| \psi^\free \|^2_{\Banach_2} \, \frac{\| K \|_\infty^{2(n+1)}}{4^{n+1}} \frac{(t_1 t_2)^{4(n+1)}}{[(2n)!]^2} \, \frac{1}{[(4n+1)(4n+2)]^2}\nonumber\\
&\leq \lambda^{2(n+1)} \, \| \psi^\free \|^2_{\Banach_2} \, \frac{\| K \|_\infty^{2(n+1)}}{4^{n+1}} \frac{(t_1 t_2)^{4(n+1)}}{[(2(n+1))!]^2}.
\end{align}
This proves \eqref{eq:2dinfind}.  \eqref{eq:2dinfind} in particular implies:
\be
	\| \varphi_n \|_{\Banach_2} \leq |\lambda|^n \, \| \psi^\free \|_{\Banach_2} \, \frac{\| K \|_\infty^n}{2^n} \frac{T^{4n}}{(2n)!} \: .
\ee
This bound shows that $\sum_i \| \varphi_i \|_{\Banach_2}$ converges. As before, we conclude that $\sum_i \varphi_i \in \Banach_2$ is the unique solution of \eqref{eq:inteq2dsimplified}. \qed
\end{proof}

\begin{theorem}[$d=3$] \label{thm:3dboundedkernel}
	For every bounded $K:\R^8 \rightarrow \C$ and every $\psi^\free \in \Banach_3$, \eqref{eq:inteq3dsimplified} possesses a unique solution $\psi \in \Banach_3$.
\end{theorem}

\begin{proof}
	The strategy of the proof and the notation are the same as before. We demonstrate that the integral operator in \eqref{eq:inteq3dsimplified} defines a bounded operator on $\Banach_3$. Then we derive an estimate for $\| \varphi_n \|_{\Banach_3}$.
	
We begin again with estimates on the $L^2$ norm of $\widehat{L} \psi$ for arbitrary $\psi\in\Banach_3$ at given times, i.e., on
\begin{multline}
\| (\widehat{L} \psi)(t_1,\cdot,t_2,\cdot)\|_{L^2}^2 = \int d^3\vx_1 \, d^3 \vx_2 \, \Biggl| \frac{\lambda}{(4\pi)^2}  \int d^3 \vx_1' \, d^3 \vx_2'~  \frac{H(t_1-|\vx_1-\vx_1'|)}{|\vx_1-\vx_1'|}\, \frac{H(t_2-|\vx_2-\vx_2'|)}{|\vx_2-\vx_2'|} \\[2mm]
 \times~ K(t_1-|\vx_1-\vx_1'|,\vx_1',t_2-|\vx_2-\vx_2'|,\vx_2') \,  \psi(t_1-|\vx_1-\vx_1'|,\vx_1',t_2-|\vx_2-\vx_2'|,\vx_2') \Biggr|^2\,.
\end{multline}
Here, we have as usual chosen an arbitrary representative of the given $\psi\in\Banach_3$. Due to the delay in the time arguments of $\psi$, it is not obvious whether $\psi$ with these arguments is even square-integrable as a function of $\vx_1'$ and $\vx_2'$, as $\psi$ was only assumed square-integrable for (almost all) fixed time arguments. So, the integral on the right-hand side might be $\infty$. Nevertheless, we can use the Cauchy-Schwarz inequality to obtain that
\begin{align}
	&\| (\widehat{L} \psi)(t_1,\cdot,t_2,\cdot)\|_{L^2}^2 \leq \frac{\lambda^2}{(4\pi)^4} \int d^3 \vx_1 \, d^3 \vx_2 \biggl[ \biggl( \int d^3 \vx_1' \,d^3 \vx_2' ~ \frac{H(t_1-|\vx_1-\vx_1'|)}{|\vx_1-\vx_1'|} \frac{H(t_2-|\vx_2-\vx_2'|)}{|\vx_2-\vx_2'|}\nonumber\\
  &\times |K|^2(t_1-|\vx_1-\vx_1'|,\vx_1',t_2-|\vx_2-\vx_2'|,\vx_2') \biggr) \biggl( \int d^3 \vx_1' \,d^3 \vx_2' ~ \frac{H(t_1-|\vx_1-\vx_1'|)}{|\vx_1-\vx_1'|} \frac{H(t_2-|\vx_2-\vx_2'|)}{|\vx_2-\vx_2'|} \nonumber\\
& \times |\psi|^2(t_1-|\vx_1-\vx_1'|,\vx_1',t_2-|\vx_2-\vx_2'|,\vx_2') \biggr) \biggr]\nonumber\\
&\leq \frac{\lambda^2}{(4\pi)^2} \| K \|_\infty^2 \frac{(t_1t_2)^2}{4} \int d^3 \vx_1 \, d^3 \vx_2 \, d^3 \vx_1' \, d^3 \vx_2'~ \frac{H(t_1-|\vx_1-\vx_1'|)}{|\vx_1-\vx_1'|} \frac{H(t_2-|\vx_2-\vx_2'|)}{|\vx_2-\vx_2'|} \nonumber\\
&\times |\psi|^2(t_1-|\vx_1-\vx_1'|,\vx_1',t_2-|\vx_2-\vx_2'|,\vx_2').
\end{align}
Now we change the order of integration, substitute the integration variables $\vx_i$ by $\vy_i = \vx_i-\vx_i'$ (the Jacobi determinant is 1), and change the order of integration back. (Since the integrand is non-negative, we can do this by virtue of Tonelli's theorem even if we do not know whether the integral is finite.) This leads to
\begin{align}
	&\| (\widehat{L} \psi)(t_1,\cdot,t_2,\cdot)\|_{L^2}^2 \leq \frac{\lambda^2}{(4\pi)^2} \| K \|_\infty^2 \frac{(t_1t_2)^2}{4} \int d^3 \vy_1 \, d^3 \vy_2 \, d^3 \vx_1' \, d^3 \vx_2'~\frac{H(t_1-|\vy_1|)}{|\vy_1|} \frac{H(t_2-|\vy_2|)}{|\vy_2|} \nonumber\\
&\times |\psi|^2(t_1-|\vy_1|,\vx_1',t_2-|\vy_2|,\vx_2') \nonumber\\
 &= \frac{\lambda^2}{(4\pi)^2} \| K \|_\infty^2 \frac{(t_1t_2)^2}{4} \int d^3 \vy_1 \, d^3 \vy_2 ~ \frac{H(t_1-|\vy_1|)}{|\vy_1|} \frac{H(t_2-|\vy_2|)}{|\vy_2|} \: \| \psi(t_1-|\vy_1|,\cdot,t_2-|\vy_2|, \cdot) \|_{L^2}^2 \nonumber\\
&= \lambda^2 \: \| K \|_\infty^2 \frac{(t_1t_2)^2}{4} \int_0^{t_1} dr_1 \, \int_0^{t_2} dr_2 ~ r_1 r_2 \, \| \psi(t_1-r_1,\cdot,t_2-r_2, \cdot) \|_{L^2}^2
\label{eq:hilfsformelind3d}
\end{align}
for all $(t_1,t_2)$. Since for almost all $(r_1,r_2)$,  $\| \psi(t_1-r_1,\cdot,t_2-r_2,\cdot) \|_{L^2} \leq \| \psi \|_{\Banach_3}$, we have that
\begin{align}
\| (\widehat{L} \psi)(t_1,\cdot,t_2,\cdot)\|_{L^2}^2 
&\leq \lambda^2 \: \| K \|_\infty^2 \frac{(t_1t_2)^2}{4} \int_0^{t_1} dr_1 \, \int_0^{t_2} dr_2 ~ r_1 r_2 \, \| \psi \|^2_{\Banach_3} \nonumber\\
&= \lambda^2 \: \| K \|_\infty^2 \frac{(t_1t_2)^2}{4} \| \psi \|^2_{\Banach_3} \frac{(t_1t_2)^2}{4} 
\label{eq:aux3ineq}
\end{align}
for all $(t_1,t_2)$ (and all representatives of $\psi$). So now we know that the left-hand side is finite.
Moreover,
\be
	\| \widehat{L} \psi \|^2_{\Banach_3} \leq \esssup_{t_1,t_2 \in [0,T]} \lambda^2 \, \| K \|_\infty^2 \frac{(t_1t_2)^4}{16} \| \psi \|^2_{\Banach_3}  = \lambda^2\, \| K \|_\infty^2 \frac{T^8}{16} \, \| \psi \|^2_{\Banach_3} \:.
\ee
This shows that $\widehat{L}$ is a bounded operator $\Banach_3 \to \Banach_3$. Hence, $\varphi_n \in  \Banach_3 ~ \forall n \in \N_0$.

We now prove the following estimate for the norm of $\varphi_n$ by induction over $n \in \N_0$:
\be
	\| \varphi_n(t_1,\cdot,t_2,\cdot) \|^2_{L^2} \leq \| \psi^\free \|^2_{\Banach_3} \: \frac{\lambda^{2n} \, \| K \|_\infty^{2n}}{4^n} \, \frac{(t_1t_2)^{4n}}{[(2n)!]^2}
	\label{eq:3dinfindl2}
\ee
for all $(t_1,t_2)$.
For $n=0$, this obviously holds. So let \eqref{eq:3dinfindl2} be true for some $n \in \N_0$. Then, plugging \eqref{eq:3dinfindl2} into \eqref{eq:hilfsformelind3d}, we obtain that
\begin{align}
	&\| \varphi_n(t_1,\cdot,t_2,\cdot)\|_{L^2}^2\nonumber\\
& \leq \| \psi^\free \|^2_{\Banach_3}  \frac{\lambda^{2(n+1)} \, \| K \|_\infty^{2(n+1)}}{4^{n+1}}  (t_1t_2)^2 \int_0^{t_1} dr_1 \, \int_0^{t_2} dr_2 ~ r_1 r_2 \frac{(t_1-r_1)^{4n}(t_2-r_2)^{4n}}{[(2n)!]^2} \nonumber\\
 &= \| \psi^\free \|^2_{\Banach_3} \frac{\lambda^{2(n+1)} \, \| K \|_\infty^{2(n+1)}}{4^{n+1}}  (t_1t_2)^2 \int_0^{t_1} d \rho_1 \, \int_0^{t_2} d \rho_2 ~ (t_1-\rho_1) (t_2-\rho_2) \frac{\rho_1^{4n} \rho_2^{4n}}{[(2n)!]^2}\nonumber\\
&= \| \psi^\free \|^2_{\Banach_3}  \frac{\lambda^{2(n+1)} \, \| K \|_\infty^{2(n+1)}}{4^{n+1}} \, \frac{(t_1t_2)^{4(n+1)}}{[(2n)!]^2} \, \frac{1}{[(4n+1)(4n+2)]^2}\nonumber\\
&\leq \| \psi^\free \|^2_{\Banach_3}  \frac{\lambda^{2(n+1)} \, \| K \|_\infty^{2(n+1)}}{4^{n+1}} \, \frac{(t_1t_2)^{4(n+1)}}{[(2(n+1))!]^2}.
\end{align}
This proves \eqref{eq:3dinfindl2}. In particular, it follows that
\be
	\| \varphi_n \|_{\Banach_3} \leq \| \psi^\free \|_{\Banach_3} \, \frac{|\lambda|^n \, \| K \|^n_\infty}{2^n} \frac{T^{4n}}{(2n)!}.
\ee
This shows that $\sum_i \| \varphi_i \|_{\Banach_3}$ converges. As before, we conclude that $\sum_i \varphi_i \in \Banach_3$ is the unique solution of \eqref{eq:inteq3dsimplified}. \qed
\end{proof}

\paragraph{Remarks.}
\begin{enumerate}
	\item Interestingly, all the main estimates are the same in dimensions $d=1,2,3$, although the integrations leading there were rather different.
	\item In a similar way as in the proofs of Thms.~\ref{thm:1dboundedkernel}-\ref{thm:3dboundedkernel}, one can show the existence and uniqueness of a solution $\psi \in L^\infty \big([0,T]^2\times \R^{2d} \big)$ (for the respective $d$ of \eqref{eq:inteq1dsimplified}-\eqref{eq:inteq3dsimplified}). In combination with Thms.~\ref{thm:1dboundedkernel}-\ref{thm:3dboundedkernel}, we then find that if $\psi^\free \in L^\infty \big([0,T]^2\times \R^{2d} \big) \cap \Banach_d$, then also $\psi \in L^\infty \big([0,T]^2\times \R^{2d} \big)\cap \Banach_d$.
\end{enumerate}

\subsection{Singular interaction kernels} \label{sec:singularkernels}

In $d=2$ and $d=3$, the physically natural interaction kernels are singular (see Eqs.~\eqref{eq:inteq2d}, \eqref{eq:inteq3d}). The main difficulty about this is that the singularities of the interaction kernel and of the Green's functions are connected. In the following, we show a possible way to deal with such connected singularities. However, compared to the physically natural cases, we still make simplifications. In $d=3$, the reason for these simplifications is that the $\delta$-functions in the interaction kernel lead to complicated weight functions. In $d=2$, the Green's functions and the interaction kernel are simply too singular in order for our strategy to work without modifications.

\paragraph{Modified singular integral equation in $d=3$.}
We consider the integral equation
\begin{align}
	\psi(t_1,\vx_1,t_2,\vx_2) = \psi^\free(t_1,\vx_1,t_2,\vx_2) + \frac{\lambda}{(4 \pi)^2} \int_0^{t_1} dt_1' \int d^3 \vx_1' \int_0^{t_2} dt_2' \int d^3 \vx_2' ~ \nonumber\\
 \frac{\delta(t_1-t_1'-|\vx_1-\vx_1'|)}{|\vx_1-\vx_1'|} \frac{\delta(t_2-t_2'-|\vx_2-\vx_2'|)}{|\vx_2-\vx_2'|} \, \frac{f(t_1',\vx_1',t_2',\vx_2')}{|\vx_1'-\vx_2'|} \, \psi(t_1',\vx_1',t_2',\vx_2'),
	\label{eq:3dsingm0}
\end{align}
where $f$ is smooth and bounded. Eq.~\eqref{eq:3dsingm0} imitates the structure of the integral equation \eqref{eq:inteq3dsimplified} for $d=3$ and $m_1=m_2=0$. The difference is that we have replaced the physically natural interaction kernel
\be
	\frac{1}{2\pi} \delta\bigl((t_1'-t_2')-|\vx_1'-\vx_2'|^2 \bigr) = \frac{1}{4 \pi\, |\vx_1'-\vx_2'|} \biggl[\delta(t_1'-t_2'-|\vx_1'-\vx_2'|) + \delta(t_1'-t_2'+|\vx_1'-\vx_2'|)\biggr]
\ee
by $f(t_1',\vx_1',t_2',\vx_2')/|\vx_1'-\vx_2'|$.

By integrating out the delta functions, \eqref{eq:3dsingm0} can be rewritten as
\begin{align}
	&\psi(t_1,\vx_1,t_2,\vx_2) = \psi^\free(t_1,\vx_1,t_2,\vx_2) + \frac{\lambda}{(4 \pi)^2} \int d^3 \vx_1'\, d^3 \vx_2' ~\frac{H(t_1-|\vx_1-\vx_1'|)}{|\vx_1-\vx_1'|} \frac{H(t_2-|\vx_2-\vx_2'|)}{|\vx_2-\vx_2'|}  \nonumber\\
 &\times \frac{f(t_1-|\vx_1-\vx_1'|,\vx_1',t_2-|\vx_2-\vx_2'|,\vx_2')}{|\vx_1'-\vx_2'|} \, \psi(t_1-|\vx_1-\vx_1'|,\vx_1',t_2-|\vx_2-\vx_2'|,\vx_2').
	\label{eq:3dsingm0b}
\end{align}

\begin{theorem} \label{thm:singularkernel3d}
	For every bounded $f:\R^8\to \C$ and every $\psi^\free \in \Banach_3$, \eqref{eq:3dsingm0b} has a unique solution $\psi \in \Banach_3$.
\end{theorem}

\begin{proof}
	The proof is structured as before. We prove that the integral operator in \eqref{eq:3dsingm0} defines a bounded operator $\widehat{L}$ on $\Banach_3$. Then we derive an estimate for $\| \varphi_n \|$.

For arbitrary $\psi \in \Banach_3$, \eqref{eq:3dsingm0b} and the Cauchy-Schwarz inequality yield that
\begin{align}
	&\| (\widehat{L} \psi)(t_1,\cdot,t_2,\cdot)\|^2_{L^2} \leq \frac{\lambda^2}{(4\pi)^4} \int d^3 \vx_1 \, d^3 \vx_2 \biggl[ \biggl( \int d^3 \vx_1' \, d^3 \vx_2' ~\frac{H(t_1-|\vx_1-\vx_1'|)}{|\vx_1-\vx_1'|^2} \frac{H(t_2-|\vx_2-\vx_2'|)}{|\vx_2-\vx_2'|^2} \nonumber\\
& \times |\psi|^2\Bigl(t_1-|\vx_1-\vx_1'|,\vx_1',t_2-|\vx_2-\vx_2'|,\vx_2'\Bigr)\biggr) \nonumber\\
&\times \biggl( \int d^3 \vx_1' \, d^3 \vx_2' ~H(t_1-|\vx_1-\vx_1'|) H(t_2-|\vx_2-\vx_2'|)\frac{|f|^2\bigl(t_1-|\vx_1-\vx_1'|,\vx_1',t_2-|\vx_2-\vx_2'|,\vx_2'\bigr)}{|\vx_1'-\vx_2'|^2} \biggr) \biggr].
\label{eq:3dsingcalc1}
\end{align}
We first consider the integral $I$ in the second round bracket and split it up into $I= I_1 + I_2$ where $I_1$ is the part with $|\vx_1'-\vx_2'| \leq 1$ and $I_2$ the part with $|\vx_1'-\vx_2'| > 1$. For the first part, we find, replacing $|f|$ with $\|f\|_\infty^2$ and leaving out the second Heaviside function:
\begin{align}
	I_1 &\leq \|f \|_\infty^2 \int\limits_{|\vx_1'-\vx_2'| \leq 1} d^3 \vx_1' \, d^3 \vx_2' ~H(t_1-|\vx_1-\vx_1'|) \, \frac{1}{|\vx_1'-\vx_2'|^2}\nonumber\\
	& =  \|f \|_\infty^2 \int d^3 \vx_1' \int_{|\vy| \leq 1} d^3 \vy ~H(t_1-|\vx_1-\vx_1'|) \,\frac{1}{|\vy|^2} \nonumber\\
 &=  \|f \|_\infty^2 \left( \int d^3 \vx_1' ~H(t_1-|\vx_1-\vx_1'|) \right) \left( \int_{|\vy| \leq 1} \frac{d^3 \vy}{|\vy|^2} \right)\nonumber\\
&=  \|f \|_\infty^2 \, \frac{4\pi t_1^3}{3} \, 4\pi \leq \|f \|_\infty^2 \,\frac{(4\pi)^2 \, T^3 }{3}.
\end{align}
For $I_2$, we obtain:
\begin{align}
	I_2 &\leq \|f \|_\infty^2 \int_{|\vx_1'-\vx_2'| > 1} d^3 \vx_1' \, d^3 \vx_2' ~H(t_1-|\vx_1-\vx_1'|) H(t_2-|\vx_2-\vx_2'|) \, \frac{1}{|\vx_1'-\vx_2'|^2}\nonumber\\
	&\leq \| f \|_\infty^2  \int d^3 \vx_1' \, d^3 \vx_2' ~H(t_1-|\vx_1-\vx_1'|) H(t_2-|\vx_2-\vx_2'|)\nonumber\\
	&= \| f \|_\infty^2  \,\left(\frac{4\pi}{3}\right)^2 \, (t_1 t_2)^3 \leq  \| f \|_\infty^2  \,\left(\frac{4\pi}{3}\right)^2 \, T^6.
\end{align}
Hence, $I \leq \|f \|^2_\infty \, \frac{(4\pi)^2}{3} (T^3 + T^6)$.

Thus, replacing the second round bracket in \eqref{eq:3dsingcalc1} by this bound for $I$, we obtain:
\begin{align}
	&\| (\widehat{L} \psi)(t_1,\cdot,t_2,\cdot)\|^2_{L^2} \leq \frac{\lambda^2}{(4\pi)^2} \, \| f \|_\infty^2  \, \frac{(T^3+T^6)}{3} \int d^3 \vx_1 \, d^3 \vx_2 \, d^3 \vx_1' \, d^3 \vx_2'~\frac{H(t_1-|\vx_1-\vx_1'|)}{|\vx_1-\vx_1'|^2}\nonumber\\
& ~~~~~\times \frac{H(t_2-|\vx_2-\vx_2'|)}{|\vx_2-\vx_2'|^2} \,  |\psi|^2(t_1-|\vx_1-\vx_1'|,\vx_1',t_2-|\vx_2-\vx_2'|,\vx_2')\nonumber\\
&= \frac{\lambda^2}{(4\pi)^2} \, \| f \|_\infty^2  \, \frac{(T^3+T^6)}{3}  \int d^3 \vy_1 \, d^3 \vy_2 \, d^3 \vx_1' \, d^3 \vx_2'~\frac{H(t_1-|\vy_1|)}{|\vy_1|^2} \frac{H(t_2-|\vy_2|)}{|\vy_2|^2} \nonumber\\
& ~~~~~\times  |\psi|^2(t_1-|\vy_1|,\vx_1',t_2-|\vy_2|,\vx_2') \nonumber\\
&=  \lambda^2 \, \| f \|_\infty^2  \, \frac{(T^3+T^6)}{3} \int_0^{t_1} dr_1 \int_{0}^{t_2} dr_2 ~ \| \psi(t_1-r_1,\cdot,t_2-r_2,\cdot) \|_{L^2}^2\nonumber\\
&= \lambda^2 \,  \| f \|_\infty^2  \, \frac{(T^3+T^6)}{3} \int_0^{t_1} d\rho_1 \int_{0}^{t_2} d\rho_2 ~ \| \psi(\rho_1,\cdot,\rho_2,\cdot) \|_{L^2}^2.
\label{eq:hilfsformel3dsing}
\end{align}
In particular, using $\| \psi(\rho_1,\cdot,\rho_2,\cdot) \|_{L^2}^2 \leq \| \psi\|^2_{\Banach_3}$ for almost every $(\rho_1,\rho_2)$, this shows that
\be
	\| \widehat{L} \psi \|_{\Banach_3} \leq |\lambda| \: \| f \|_\infty \, \frac{(T^5+T^8)^{1/2}}{\sqrt{3}} \,  \|\psi \|_{\Banach_3} \:.  
\ee
Hence, the integral operator $\widehat{L}$ is bounded, and $\varphi_n \in \Banach_3\, \forall n \in \N_0$.

We now turn to the estimate of $\| \varphi_n(t_1,\cdot,t_2,\cdot) \|_{L^2}$. We shall prove by induction over $n \in \N_0$:
\be
	\| \varphi_n(t_1,\cdot,t_2,\cdot) \|_{L^2}^2 \leq \left( \lambda^2 \,  \| f \|_\infty^2  \, \frac{(T^3+T^6)}{3} \right)^n \, \frac{(t_1t_2)^n}{(n!)^2} \, \| \psi^\free \|^2_{\Banach_3}.
	\label{eq:3dsingind}
\ee
For $n=0$, this obviously holds. So let \eqref{eq:3dsingind} be true for some $n \in \N_0$. Plugging \eqref{eq:3dsingind} into \eqref{eq:hilfsformel3dsing} for $\psi = \varphi_n$, we find:
\begin{align}
	&\| \varphi_{n+1}(t_1,\cdot,t_2,\cdot) \|_{L^2}^2 \leq \left(\lambda^2 \,  \| f \|_\infty^2  \, \frac{(T^3+T^6)}{3} \right)^{n+1} \int_0^{t_1} d\rho_1 \int_{0}^{t_2} d\rho_2 ~ \frac{(\rho_1 \rho_2)^n}{(n!)^2}\, \| \psi^\free \|^2\nonumber\\
	&=  \left( \lambda^2 \,  \| f \|_\infty^2  \, \frac{(T^3+T^6)}{3}\right)^{n+1} \frac{(t_1 t_2)^{n+1}}{((n+1)!)^2}\, \| \psi^\free \|^2.
\end{align}
This proves \eqref{eq:3dsingind}. In particular,  \eqref{eq:3dsingind} implies:
\be
	\| \varphi_n \|_{\Banach_3} \leq |\lambda| \:  \| f \|_\infty  \, \frac{(T^3+T^6)^{1/2}}{\sqrt{3}} \: \frac{T^n}{n!} \: \| \psi^\free \|_{\Banach_3}.
\ee
This shows that $\sum_i \| \varphi_i \|$ converges. As before, we conclude that $\sum_i \varphi_i \in \Banach_3$ is the unique solution of \eqref{eq:3dsingm0b}. \qed
\end{proof}

\begin{remark}
	Splitting the singularities of $G_1,G_2$ and $K$ via the Cauchy-Schwarz inequality does not work for the physically natural equation \eqref{eq:inteq2d} in $d=2$. The reason is that the integral $\int_0^t dt' \int d^2 \vx' ~\frac{H(t'-|\vx'|)}{(t')^2-|\vx'|^2}$ diverges. However, we can treat a problem with $[(t')^2-|\vx'|^2]^{-\alpha/2}$ with $\alpha < 1$ instead of $[(t')^2-|\vx'|^2]^{-1/2}$.
\end{remark}

\paragraph{Modified singular integral equation in $d=2$.}

Let $0 < \alpha < 1$. The previous remark suggests to consider the following integral equation on $\Banach_2$:
\begin{align}
	&\psi(t_1,\vx_1,t_2,\vx_2) = \psi^\free(t_1,\vx_1,t_2,\vx_2) + \frac{\lambda}{(2\pi)^3} \int_0^{t_1} dt_1' \int d^2 \vx_1' \int_0^{t_2} dt_2' \int d^2 \vx_2'~\nonumber\\
& \frac{H(t_1-t_1'-|\vx_1-\vx_1'|)}{[(t_1-t_1')^2 - |\vx_1-\vx_1'|^2]^{\alpha/2}} \cos\left( m_1\sqrt{(t_1-t_1')^2-|\vx_1-\vx_1'|^2}\right) \frac{H(t_2-t_2'-|\vx_2-\vx_2'|)}{[(t_2-t_2')^2 - |\vx_2-\vx_2'|^2]^{\alpha/2}} \nonumber\\
& \times  \cos\left( m_2\sqrt{(t_2-t_2')^2-|\vx_2-\vx_2'|^2}\right) \frac{H((t_1'-t_2')^2-|\vx_1'-\vx_2'|^2)}{[(t_1'-t_2')^2 - |\vx_1'-\vx_2'|^2]^{\alpha/2}}\, \psi(t_1',\vx_1',t_2',\vx_2').
\label{eq:2dsingalpha}
\end{align}

\begin{theorem} \label{thm:singularkernel2d}
	For every $\psi^\free \in \Banach_2$, \eqref{eq:2dsingalpha} has a unique solution $\psi \in \Banach_2$.
\end{theorem}

\begin{proof}
	The proof is structured like the previous ones. First we show that the integral operator $\widehat{L}$ in \eqref{eq:2dsingalpha} is a bounded operator on $\Banach_2$. Then we derive an estimate for the norm of $\varphi_n$ (defined analogously as before).

For the boundedness, we use \eqref{eq:2dsingalpha} and the Cauchy-Schwarz inequality to obtain:
\begin{align}
	&\| \widehat{L} \psi(t_1,\cdot,t_2,\cdot) \|^2_{L^2} \leq \frac{\lambda^2}{(2\pi)^6} \int d^2 \vx_1 \, d^2 \vx_2 \left[ \left( \int_0^{t_1} dt_1'  \int_0^{t_2} dt_2' \int d^2 \vx_1' \, d^2 \vx_2'~\right. \right. \nonumber\\
& \frac{H(t_1-t_1'-|\vx_1-\vx_1'|)}{[(t_1-t_1')^2 - |\vx_1-\vx_1'|^2]^\alpha} \left. \frac{H(t_2-t_2'-|\vx_2-\vx_2'|)}{[(t_2-t_2')^2 - |\vx_2-\vx_2'|^2]^\alpha} |\psi|^2(t_1',\vx_1',t_2',\vx_2')\right)\nonumber\\
&\times \left( \int_0^{t_1} dt_1'  \int_0^{t_2} dt_2' \int d^2 \vx_1' \, d^2 \vx_2'~ H(t_1-t_1'-|\vx_1-\vx_1'|) H(t_2-t_2'-|\vx_2-\vx_2'|) \right. \nonumber\\
& \left. \left. \times \frac{H((t_1'-t_2')^2-|\vx_1'-\vx_2'|^2)}{[(t_1'-t_2')^2 - |\vx_1'-\vx_2'|^2]^\alpha} \right) \right]
	\label{eq:2dsingalphacalc1}
\end{align}
We first estimate the expression in the second round bracket. 
Changing variables, $\vx_i' \rightarrow \vy_i = \vx_i'-\vx_i$, it becomes:
\be
	\int_0^{t_1} dt_1'  \int_0^{t_2} dt_2' \int d^2 \vy_1 \, d^2 \vy_2~ H(t_1-t_1'-|\vy_1|) H(t_2-t_2'-|\vy_2|) \frac{H((t_1'-t_2')^2-|\vy_1+\vx_1-\vy_2-\vx_2|^2)}{[(t_1'-t_2')^2 - |\vy_1+\vx_1-\vy_2-\vx_2|^2]^\alpha}
\ee
Changing variables another time, namely to $\vy = \vy_1+\vx_1-\vy_2-\vx_2$ and $\mathbf{Y} = \vy_1 + \vy_2$ (with the Jacobi determinant $\frac{1}{4}$) and then introducing spherical coordinates for $\vy$ and $\mathbf{Y}$, we see that the expression is smaller than or equal to
\begin{align}
	&\frac{(2\pi)^2}{4} \int_0^{t_1} dt_1'  \int_0^{t_2} dt_2' \int_0^{t_1-t_1'+t_2 - t_2'} dR \, R \int_0^{|t_1'-t_2'|} dr ~ \frac{r}{((t_1'-t_2')^2-r^2)^{\alpha}}\nonumber\\
&= \frac{(2\pi)^2}{4} \int_0^{t_1} dt_1'  \int_0^{t_2} dt_2' \int_0^{t_1-t_1'+t_2 - t_2'} dR \, R ~ \left[- \frac{1}{2(1-\alpha)}((t_1'-t_2')^2-r^2)^{1-\alpha}\right]_{0}^{|t_1'-t_2'|}\nonumber\\
&= \frac{(2\pi)^2}{4} \int_0^{t_1} dt_1'  \int_0^{t_2} dt_2' \int_0^{t_1-t_1'+t_2 - t_2'} dR \, R ~\frac{1}{2(1-\alpha)}|t_1'-t_2'|^{2(1-\alpha)}\nonumber\\
&= \frac{(2\pi)^2}{2^4(1-\alpha)} \int_0^{t_1} dt_1'  \int_0^{t_2} dt_2' ~(t_1-t_1'+t_2 - t_2')^2 \, |t_1'-t_2'|^{2(1-\alpha)}.
\end{align}
For our purposes, a crude upper bound is sufficient. To this end, we note $(t_1-t_1'+t_2 - t_2')^2 \leq (t_1+t_2)^2$ and $|t_1'-t_2'|^{2(1-\alpha)} \leq (t_1+t_2)^{2(1-\alpha)}$. Thus, we see that the previous expression is smaller than or equal to
\be
	\frac{(2\pi)^2}{2^4(1-\alpha)} (t_1+t_2)^2 \: (t_1+t_2)^{2(1-\alpha)} \: t_1 t_2 \leq  \frac{(2\pi)^2}{2^4(1-\alpha)} (t_1+t_2)^{6-2\alpha}.
\ee
With this result, \eqref{eq:2dsingalphacalc1} becomes:
\begin{align}
	&\| \widehat{L} \psi(t_1,\cdot,t_2,\cdot) \|^2_{L^2} \leq \lambda^2 \frac{(t_1+t_2)^{6-2\alpha}}{(2\pi)^4 \cdot 2^4(1-\alpha)}  \int_0^{t_1} dt_1'  \int_0^{t_2} dt_2' \int d^2 \vx_1 \, d^2 \vx_2  \, d^2 \vx_1' \, d^2 \vx_2' \nonumber\\
& \frac{H(t_1-t_1'-|\vx_1-\vx_1'|)}{[(t_1-t_1')^2 - |\vx_1-\vx_1'|^2]^\alpha} \frac{H(t_2-t_2'-|\vx_2-\vx_2'|)}{[(t_2-t_2')^2 - |\vx_2-\vx_2'|^2]^\alpha} |\psi|^2(t_1',\vx_1',t_2',\vx_2').
\end{align}
We now change variables, $\vx_i \rightarrow \vx_i-\vx_i' =: \vy_i$ (Jacobi determinant 1). This yields:
\begin{align}
	&\| \widehat{L} \psi(t_1,\cdot,t_2,\cdot) \|^2_{L^2} \leq \lambda^2 \frac{(t_1+t_2)^{6-2\alpha}}{(2\pi)^4 \cdot 2^4(1-\alpha)}  \int_0^{t_1} dt_1'  \int_0^{t_2} dt_2' \int d^2 \vy_1 \, d^2 \vy_2  \, d^2 \vx_1' \, d^2 \vx_2' \nonumber\\
&~~~ \frac{H(t_1-t_1'-|\vy_1|)}{[(t_1-t_1')^2 - |\vy_1|^2]^\alpha} \frac{H(t_2-t_2'-|\vy_2|)}{[(t_2-t_2')^2 - |\vy_2|^2]^\alpha} |\psi|^2(t_1',\vx_1',t_2',\vx_2') \nonumber\\
&= \lambda^2 \frac{(t_1+t_2)^{6-2\alpha}}{(2\pi)^4 \cdot 2^4(1-\alpha)}  \int_0^{t_1} dt_1'  \int_0^{t_2} dt_2' \int d^2 \vy_1 \, d^2 \vy_2 \nonumber\\
&~~~\frac{H(t_1-t_1'-|\vy_1|)}{[(t_1-t_1')^2 - |\vy_1|^2]^\alpha} \frac{H(t_2-t_2'-|\vy_2|)}{[(t_2-t_2')^2 - |\vy_2|^2]^\alpha} \, \| \psi(t_1',\cdot,t_2',\cdot)\|^2_{L^2}\nonumber\\
&= \lambda^2 \frac{(t_1+t_2)^{6-2\alpha}}{(2\pi)^2 \cdot 2^6(1-\alpha)^3}  \int_0^{t_1} dt_1'  \int_0^{t_2} dt_2'~ (t_1-t_1')^{2(1-\alpha)}(t_2-t_2')^{2(1-\alpha)} \| \psi(t_1',\cdot,t_2',\cdot)\|^2_{L^2}\nonumber\\
&\leq \lambda^2 \frac{(2T)^{10-4\alpha}}{(2\pi)^2 \cdot 2^6(1-\alpha)^3} \int_0^{t_1} dt_1'  \int_0^{t_2} dt_2'~\| \psi(t_1',\cdot,t_2',\cdot)\|^2_{L^2}.
\label{eq:2dsingalphahilfsformel}
\end{align}
Using $\| \psi(t_1',\cdot,t_2',\cdot)\|^2_{L^2} \leq \| \psi \|^2$, we deduce:
\begin{align}
	\| \widehat{L} \psi \|^2 \leq \lambda^2 \frac{(2T)^{12-4\alpha}}{(2\pi)^2 \cdot 2^6(1-\alpha)^3} \, \| \psi \|^2.
\end{align}
This shows that $\widehat{L}$ indeed is a bounded operator on $\Banach_2$. Next, we prove the following estimate by induction over $n \in \N_0$:
\be
	\| \varphi_n(t_1,\cdot,t_2,\cdot) \|^2_{L^2} \leq \| \psi^\free\|^2 \, \left( \frac{\lambda^2 \cdot (2T)^{12-4\alpha}}{(2\pi)^2 \cdot 2^6(1-\alpha)^3} \right)^n \, \frac{(t_1t_2)^n}{(n!)^2}.
	\label{eq:2dsingalphaind}
\ee
For $n=0$, \eqref{eq:2dsingalphaind} obviously holds. So let  \eqref{eq:2dsingalphaind} be true for some $n \in \N_0$. Now $\varphi_{n+1} = \widehat{L} \varphi_n$ and thus plugging \eqref{eq:2dsingalphaind} into \eqref{eq:2dsingalphahilfsformel} yields:
\begin{align}
	\| \widehat{L} \varphi_{n+1}(t_1,\cdot,t_2,\cdot) \|^2_{L^2} &\leq \| \psi^\free \|^2 \, \left( \frac{\lambda^2 \cdot (2T)^{12-4\alpha}}{(2\pi)^2 \cdot 2^6(1-\alpha)^3} \right)^{n+1} \int_0^{t_1} dt_1'  \int_0^{t_2} dt_2'~\frac{(t_1't_2')^n}{(n!)^2}\nonumber\\
&= \| \psi^\free \|^2 \, \left( \frac{\lambda^2 \cdot (2T)^{12-4\alpha}}{(2\pi)^2 \cdot 2^6(1-\alpha)^3} \right)^{n+1} \, \frac{(t_1t_2)^{n+1}}{[(n+1)!]^2}.
\end{align}
This proves \eqref{eq:2dsingalphaind}. In particular, \eqref{eq:2dsingalphaind} implies:
\be
	\| \varphi_n \|_{\Banach_2} \leq  \frac{\| \psi^\free\|_{\Banach_2}}{n!} \, \left( \frac{\lambda^2 \cdot (2T)^{14-4\alpha}}{(2\pi)^2 \cdot 2^6(1-\alpha)^3} \right)^{n/2}.
\ee
This bound shows that $\sum_i \| \varphi_i \|$ converges. Hence $\sum_i \varphi_i$ converges in $\Banach_2$, and analogously to before it follows that this series is, in fact, the unique solution of \eqref{eq:2dsingalpha}. \qed	
\end{proof}

\section{Conclusions} \label{sec:conclusion}

\paragraph{Summary.}
In this paper, we have provided the first existence and uniqueness results for certain classes of multi-time integral equations of the form \eqref{eq:inteq}. We have focused on the case that the Green's functions in \eqref{eq:inteq} are retarded Green's functions of the Klein-Gordon equation. It was then demonstrated that assuming a beginning of time (which seems plausible in view of the Big Bang of our universe), these integral equations attain a Volterra-type structure. Hence, the time integrals run only from 0 to $t_i$, $i=1,2$. However, compared to the standard cases for multi-dimensional Volterra integral equations, the multi-time integral equations in this paper show several new features that necessitate a novel treatment: (a)~combined space and time integrals occur, (b)~the integral kernels are not square integrable, (c)~the kernels are singular (except for $d=1$). This singular behavior manifests itself in two aspects:  (i) the Green's functions are singular and (ii) the interaction kernel is singular as well. The fact that these two types of singularities are connected makes them particularly challenging.

We were able to give results covering each of these features (a)--(c) for integral equations that are simplified compared to the physically natural cases. (In $d=1$, however, the physically natural case is covered.) The simplifications were introduced with care in order not to deviate overly from these natural cases. In particular, arbitrary but bounded interaction kernels have been covered. We also proved some results for singular interaction kernels which, in a certain sense, approximate the physically natural ones. This shows that it is, in principle, possible to deal with the above-mentioned connected singularities.

\paragraph{Discussion.}

In the context of other equations involving time delay, such as delay differential equations, our results may appear surprising. For these equations, it is notoriously hard to prove the existence and uniqueness of solutions. However, our equations involve time delay (even in many variables) and we have obtained global existence and uniqueness results for them. So the question arises: what makes our equations more tractable than a typical delay differential equation?
We believe that mainly two features are responsible: 1. our equations are linear while the typical delay differential equation is not, and 2., in our case, the time delay is bounded because of the assumed beginning in time and an arbitrary final time $T$ up to which we would like to solve the equation.

Our results also shed new light on a question discussed in \cite{direct_interaction_quantum}, namely: which data parametrize the solution spaces of multi-time integral equations \eqref{eq:inteq}? The conjecture that a given solution $\psi^\free$ of the free multi-time equations determines the solution $\psi$ of \eqref{eq:inteq} uniquely has turned out correct for the multi-time equations studied in this paper. We can even say more than that. As a consequence of the beginning in time and the retarded Green's functions, $\psi^\free(t_1=0,\cdot,t_2=0,\cdot)$ plays the role of initial data for $\psi$; that is, we obtain a Cauchy problem ``at the Big Bang.''

\paragraph{Outlook.} There are several interesting questions which have been left open by our work, or have been opened up by it:
\begin{enumerate}
	\item It would be desirable to treat the physically natural singular integral kernels also in $d=2,3$ (Eqs. \eqref{eq:inteq2d}, \eqref{eq:inteq3d}). It is a challenge to find a proof or disproof of existence and uniqueness of solutions of these equations. This likely requires a modification of our techniques.
	\item In the present paper, the Big Bang was only taken as a reason to introduce a lower limit for the time integrals. It would be desirable to implement it in a physically natural way instead. This would mean to formulate the integral equation on curved spacetimes with a Big Bang singularity and requires in particular to explicitly determine the Green's functions on curved spacetimes. Furthermore, we expect additional singularities of the Green's functions to appear on spacetimes with a Big Bang, as a consequence of the latter. We address this circle of questions in a subsequent paper \cite{int_eq_curved}.
	\item Physically, it would be more natural to study the case of Green's function of the Dirac equation instead of the Klein-Gordon equation. While the KG equation is normally only used as a toy model, the Dirac equation describes actual elementary particles, e.g., electrons. In the Dirac case, the Green's functions become more singular (they involve $\delta'$-functions).
	\item Finally, the case of time-symmetric Green's functions would be of great interest as the integral equation \eqref{eq:inteq} then is time reversal invariant, a property which is usually expected from fundamental physical laws. In this case, the equation does not have a  Volterra structure any more, and it becomes much harder to derive existence and uniqueness results. A beginning in time alone does not simplify the problem much; one would also need an end in time. In fact, this is also a possible cosmological scenario: the Big Crunch. It would be of interest to develop existence results for this case (see \cite{int_eq_curved}).
\end{enumerate}

\paragraph{Acknowledgments.}
We would like to thank Markus N\"oth and Shadi Tahvildar-Zadeh for helpful discussions. Special thanks go to Fioralba Cakoni for valuable advice.\\[1.5mm]
\begin{minipage}{15mm}
\includegraphics[width=13mm]{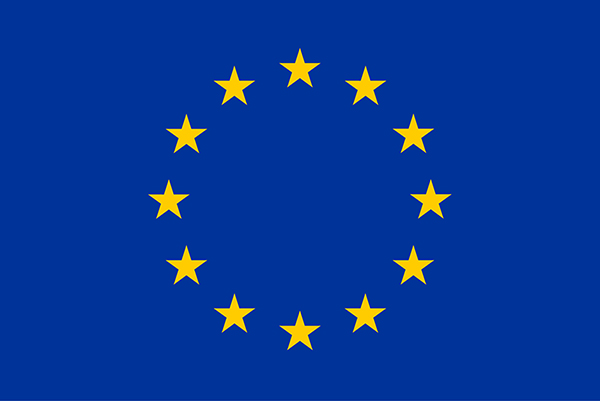}
\end{minipage}
\begin{minipage}{143mm}
This project has received funding from the European Union's Framework for Re-\\
search and Innovation Horizon 2020 (2014--2020) under the Marie Sk{\l}odowska-
\end{minipage}\\[1mm]
Curie Grant Agreement No.~705295.

\end{document}